\documentclass[11pt]{article}

\usepackage[margin=1in]{geometry}

\usepackage{amsmath,amssymb,amsfonts,amsthm}

\usepackage{graphicx}
\usepackage{subcaption}
\usepackage{float}
\usepackage[section]{placeins}
\usepackage{booktabs}

\usepackage{algorithm}
\usepackage{algpseudocode}

\usepackage{cite}
\usepackage{pifont}
\usepackage{textcomp}
\usepackage{xcolor}
\usepackage{hyperref}
\hypersetup{
    colorlinks=true,
    linkcolor=blue,
    citecolor=blue,
    urlcolor=blue
}

\allowdisplaybreaks

\newtheorem{problem}{Problem}
\newtheorem{lemma}{Lemma}
\newtheorem{corollary}{Corollary}
\newtheorem{remark}{Remark}
\newtheorem{definition}{Definition}
\newtheorem{theorem}{Theorem}
\newtheorem{example}{Example}
\newtheorem{proposition}{Proposition}
\newtheorem{assumption}{Assumption}

\def\BibTeX{{\rm B\kern-.05em{\sc i\kern-.025em b}\kern-.08em
    T\kern-.1667em\lower.7ex\hbox{E}\kern-.125emX}}
\begin{document}
\title{A Data-Enabled Primal-Dual Approach for Policy Learning with SDP Formulations}
\author{
\begin{tabular}{ccc}
Han Wang & Feiran Zhao & Florian D\"orfler
\end{tabular}\\[1mm]
Automatic Control Laboratory (IfA), ETH Z\"urich\\
\texttt{\{hanwang1,zhaofe,dorfler\}@ethz.ch}
}

\date{}
\maketitle

\begin{abstract}
This paper develops a data-enabled primal-dual framework for learning optimal control policies for unknown linear discrete-time systems from online data. The proposed approach views the data-dependent control synthesis problem as a time-varying semidefinite program (SDP) whose coefficients are recursively updated from online closed-loop measurements. Instead of repeatedly solving a full SDP as new data arrive, the policy is updated online through lightweight primal-dual iterations, each consisting of a linear equation solve and a projection onto the positive semidefinite cone. The framework applies to both direct and indirect data-driven formulations and covers a broad class of control objectives, including LQR, $\mathcal{H}_\infty$ control, and safety-critical control. To characterize the coupling between online optimization and closed-loop data generation, we introduce two data-dependent quantities: the Sim-to-Real Gap, which measures the mismatch between noisy and noiseless data-induced SDPs, and the Difference-of-Signal, which measures the temporal variation of the SDP coefficients. Under persistency of excitation, suitable SDP regularity conditions, and sufficiently slow data variation, we establish a local linear tracking result up to residual terms governed by the latter two quantities. A global ergodic convergence bound is also derived for arbitrary initialization. Numerical examples on LQR, $\mathcal{H}_\infty$ control, and safe exploration demonstrate that the proposed method can efficiently improve control performance from online data while accommodating SDP constraints beyond the well-explored LQR policy-gradient formulations.
\end{abstract}


\section{Introduction}
\label{sec:introduction}
Data-driven control aims to synthesize controllers from measured trajectories, thereby reducing the dependence on explicit first-principles modeling. This paradigm is particularly attractive for modern control applications, where high-dimensional dynamics, uncertain environments, and repeated operation make it difficult to obtain an accurate model beforehand, while closed-loop data are continuously generated during deployment. A central challenge is therefore not only how to design a controller from a fixed batch of data, but also how to exploit newly collected data to recursively improve and adapt the learned control policy online.

\subsection{Literature review}

A key theoretical foundation of data-driven control is that measured trajectories can serve as an implicit representation of unknown dynamics, as formalized by Willems' fundamental lemma under persistency of excitation \cite{willems2005note}. This viewpoint has supported both predictive-control formulations such as DeePC \cite{coulson2019data,coulson2019regularized,berberich2021data} and data-dependent linear matrix inequalities methods for stabilization, LQR, and robust control \cite{de2019formulas,berberich2020robust,dorfler2023certainty}, with related informativity results clarifying when a dataset is sufficient to certify control properties \cite{van2020data}. These methods are typically used in a batch manner: a controller is synthesized from a fixed dataset, while data collected during subsequent operation are not directly used to refine the policy. This motivates a recursive learning viewpoint, where the policy is updated as new closed-loop data arrive, without repeatedly solving the full data-dependent synthesis problem from scratch.

Recent online policy-learning methods partially address this issue by combining closed-loop data with recursive controller updates, especially for LQR problems \cite{dean2018regret,rantzer2021minimax,zhao2025data,zhao2025policy}. These methods provide an appealing mechanism for improving optimal control policies from online data and have also been validated in robotic applications \cite{persson2026adaptive}. Their algorithmic structure, however, is often tailored to LQR-type objectives for which policy-gradient or policy-improvement steps admit tractable expressions. For more general specifications, such as $\mathcal H_{\infty}$ control or safety-critical control with SDP constraints, such closed-form updates are not directly available \cite{zheng2026benign}. This motivates an alternative route. Instead of deriving a policy gradient formula for each specific control objective, we treat the corresponding data-dependent SDP itself as the object to be learned. As new closed-loop data update the coefficients of this SDP, the control policy can be updated by recursively tracking its time-varying optimizer. 

This perspective is closely related to online and time-varying optimization, where decisions are updated sequentially as new information becomes available. Classical online convex optimization studies regret with respect to a fixed comparator \cite{zinkevich2003online}, while dynamic settings require tracking time-varying optimizers and lead to bounds depending on path length, temporal variation, or related measures of nonstationarity \cite{hall2015online,zhang2018adaptive,zhao2020dynamic}. For structured constrained problems, online alternating direction methods \cite{wang2012online}, online proximal-ADMM \cite{zhang2021online}, and more general time-varying convex optimization schemes \cite{simonetto2020time} provide primal-dual mechanisms for updating decisions with limited per-step computation. {Such approaches are particularly attractive in control applications, where repeatedly solving a full SDP to optimality can be computationally demanding~\cite{helmberg2000semidefinite}, and standard interior-point methods are generally difficult to warm-start \cite{skajaa2015warmstarting}. Moreover, exact solutions of data-dependent SDPs may become overly sensitive to noisy and continuously changing data, whereas incremental primal-dual updates naturally provide a gradual adaptation mechanism that is more compatible with maintaining closed-loop robust stability.}

The data-enabled policy-learning problem considered here differs from standard online optimization in an important way. In both direct and indirect data-driven parameterizations, the time variation of the underlying SDP is not specified by an exogenous sequence of costs or constraints, but is generated by the closed-loop trajectory through recursively updated data matrices. In the direct setting, the SDP coefficients depend explicitly on online trajectory data, whereas in the indirect setting they evolve through continuously updated system estimates constructed from the data. Consequently, the optimizer drift is coupled with the plant dynamics, probing signal, noise level, and persistency of excitation. This calls for a convergence analysis that links online primal-dual optimization with data-dependent quantities specific to control. In this paper, we develop such a framework by introducing a data-enabled primal-dual method for tracking time-varying SDPs through a small number of online primal-dual iterations rather than repeatedly solving each SDP to full optimality, and by quantifying its convergence through the Sim-to-Real Gap and Difference-of-Signal, which capture the quality and temporal evolution of the online data, respectively.

\subsection{Contributions}

To bridge the gap, we propose a unified data-enabled primal-dual framework for recursively learning optimal control policies through SDP formulations. The framework encompasses both direct and indirect data-driven parameterizations by viewing the resulting data-dependent control SDP as a time-varying optimization problem whose constraint matrices are updated from closed-loop data. Rather than repeatedly solving each SDP to full optimality, the proposed method tracks its solution online using lightweight primal-dual iterations. Our contributions are summarized as follows:

\begin{itemize}
    \item We develop a recursive policy-learning framework for optimal control problems admitting SDP formulations. Unlike existing online data-driven policy-search methods that are mainly tailored to the LQR objective, the proposed data-enabled primal-dual method applies to a broader class of control problems, including LQR, $\mathcal H_{\infty}$ optimal control, and safety-critical control with state constraints. The method therefore provides a unified online optimization viewpoint for learning optimal control policies from data.

    \item We propose a computationally efficient online algorithm for tracking the solution of time-varying data-driven SDPs. At each time step, newly collected closed-loop data update the SDP coefficients, and the policy is updated through primal-dual iterations involving only a linear equation solve and a projection onto the positive semidefinite cone. This avoids repeatedly solving a full SDP from scratch and makes the approach suitable for recursive implementation.

    \item We provide convergence analysis for both direct and indirect data-driven formulations. The analysis identifies two key quantities governing online performance: the sim-to-real gap, which quantifies the accuracy of the data-dependent representation, and the difference-of-signal, which measures how fast the underlying SDP changes as new data arrive. These quantities explicitly connect the convergence behavior of the primal-dual algorithm to the quality and evolution of the closed-loop data.

    \item We establish both local and global convergence guarantees for the resulting time-varying SDP tracking problem. Under suitable SDP regularity conditions, persistency of excitation, and sufficiently small DOS, we prove a local linear tracking result showing that the primal-dual iterates approach the time-varying optimal solution up to residual terms that scale linearly with DOS and SRG. For more general initializations, we further derive a global ergodic convergence bound.

    \item We demonstrate the effectiveness of the proposed method on several control tasks. For the LQR benchmark, our method achieves performance comparable to state-of-the-art data-enabled policy optimization methods \cite{zhao2025data}, while exhibiting robustness when the data are corrupted by large process noise. Beyond LQR, we validate the method on a $14$-dimensional $\mathcal H_{\infty}$ control problem, which is challenging for existing policy-gradient-based approaches. We further show that the same framework can accommodate safety-critical control problems where state constraints are imposed through SDP conditions.
\end{itemize}

Section \ref{sec:formulation} provides the problem formulation. The main algorithm is presented in Section \ref{sec:algorithm}. The convergence analysis is conducted in Section \ref{sec:analysis}. Section \ref{sec:simulation} presents simulation results. All proofs are given in the Appendix.

\subsection{Notation}

For a matrix $M$, $\|M\|_F$, $\|M\|_2$, and $\underline{\sigma}(M)$ denote its Frobenius norm, induced spectral norm, and smallest singular value, respectively. The column-wise vectorization of $M$ is denoted by $\operatorname{vec}(M)$, and $\operatorname{vec}^{-1}(\cdot)$ denotes the inverse operation with dimension clear from the context. The projection onto the positive semidefinite cone is denoted by $\Pi_{\mathbb{S}_+}(\cdot)$. The indicator function of a set $\mathcal C$ is denoted by $\mathcal I_{\mathcal C}(\cdot)$. For a vector $\omega\in \mathbb{R}^\omega$ and a scalar $r>0$, $\mathcal{B}(\omega,r):=\{x\in \mathbb{R}^\omega:\|x-\omega\|_F\le r\}$.

\section{Problem Formulation}\label{sec:formulation}

Consider the following linear discrete-time system:
\begin{equation}\label{eq:system}
    x_{t+1}=Ax_t+Bu_t+w_t,
\end{equation}
where $t\in\mathbb{N}$, $x_t\in\mathbb{R}^n$ denotes the state at time $t$, $u_t\in\mathbb{R}^m$ is the control input, and $w_t\in\mathbb{R}^n$ is the process noise. We assume that the system matrices $A$ and $B$ are unknown. Instead, input and state measurements of the system is assumed available.

For the system \eqref{eq:system}, a sequence of states, inputs, noises and successor states measures is defined by:
\begin{equation}\label{eq:data}
    \begin{aligned}
        X_{0,t}&:=[x_0\quad x_1\quad \ldots\quad x_{t-1}]\in\mathbb{R}^{n\times t}\\
        U_{0,t}&:=[u_0\quad u_1\quad \ldots \quad u_{t-1}]\in\mathbb{R}^{m\times t}\\
        W_{0,t}&:=[w_0\quad w_1\quad \ldots \quad w_{t-1}]\in\mathbb{R}^{n\times t}\\
        X_{1,t}&:=[x_1\quad x_2\quad \ldots \quad x_{t}]\in\mathbb{R}^{n\times t},\\
        D_t&:=[X_{0,t}\quad X_{1,t}\quad U_{0,t}].
    \end{aligned}
\end{equation}
Note that each column of the data matrices do not need to be consecutive in time. It is only required that the elements that are in the same columns of $X_{0,t_0},U_{0,t_0},W_{0,t_0}$ are sampled at the same time, and that in $X_1$ be the successor state.
These matrices satisfy the subspace relation
\begin{equation}\label{eq:data-relation}
    X_{1,t}=AX_{0,t}+BU_{0,t}+W_{0,t}.
\end{equation}

{The goal is to learn the optimal linear control policy $u^*_t=K^*x_t$ from the state and input measurements $D_t$, where the optimal control gain $K^*\in \mathbb{R}^{m\times n}$ is the optimal solution of a policy optimization problem. Take the LQR as an example, the problem is given by \cite{dorfler2023certainty}:
\begin{equation}\label{eq:po}
    \begin{aligned}
        [K^*,\Sigma^*]=\mathop{\arg\min}_{K,\Sigma}~&\mathrm{Tr}[(Q+K^\top RK)\Sigma]\\
        \mathrm{subject~to}~&\Sigma = I_n+(A+BK)\Sigma (A+BK)^\top 
    \end{aligned}
\end{equation}
}
where $\Sigma\in \mathbb{R}^{n\times n}$ is the Gramian matrix. 

Two challenges are exhibited in the learning problem: i) As the noise $W_{0,t}$ is unknown and corrupting the estimate of the system matrices $(A,B)$, it is usually unlikely to learn the optimal gain $K^*$ with finite measures $t<\infty$; ii) the policy optimization problem is usually nonconvex. Looking at \eqref{eq:po} both the cost function and the constraint set are nonconvex.

To address these challenges, we propose a learning scheme that consists of two main strategies:\\ 
\textbf{Convex reformulation}: Instead of working on the nonconvex policy optimization problem, we reformulate it into a convex semi-definite program. {Take the LQR problem \eqref{eq:po} as an example, the SDP formulation is given by \cite{dorfler2023certainty}:
\begin{equation}\label{eq:model-lqr}
    \begin{aligned}
        [\Sigma^*,Y^*,L^*]=\mathop{\arg\min}_{\Sigma,Y,L}~&\mathrm{Tr}(Q\Sigma)+\mathrm{Tr}(RL)\\
        \mathrm{subject~to}~&\begin{bmatrix}
            \Sigma-I_n&A\Sigma+BY\\\star &\Sigma
        \end{bmatrix}\succeq 0\\
        &\begin{bmatrix}
            L&Y\\Y^\top& \Sigma
        \end{bmatrix}\succeq 0
    \end{aligned}
\end{equation}
and the optimal control gain is $K^*=Y^*(\Sigma^*)^{-1}$.
}

It is known that a large amount of optimal control problems for linear systems, including LQR, $\mathcal{H}_\infty$ and safety-critical control, can be formulated as SDPs \cite{boyd1994linear}.\\
\textbf{Policy Learning}: Denote the learned policy at time $t$ by $K_t$. Our goal is to design a policy learning mechanism $\mathcal{P}$ such that
\begin{equation}\label{eq:learning-mechanism}
    K_t=\mathcal{P}(K_{t-1};D_{t}),\quad t=t_0,\ldots
\end{equation}
where $K_{t_0}$ denotes the initialized control gain with offline data $D_{t_0}$. The system runs in closed-loop as $x_{t+1}=Ax_t+Bu_t+w_t$, where $u_t=K_tx_t+e_t$. The probing noise $e_t\in\mathbb{R}^m$ is introduced to enhance exploration, while ensuring  persistent of excitation \cite{faradonbeh2020input}. Then, the set of data consists of two parts at each time $t$: i) offline data ${X_{0,t_0}, U_{0,t_0}, X_{1,t_0}}$ and ii) online closed-loop data ${x_t, u_t}$. The desired outcome of the adaptive control mechanism is that, the learned gain $K_t$ converges to the optimal one $K^*$ as time evolves and in presence of disturbances.

The key problems in the above learning scheme are formalized as follows.

{\begin{problem}[Recursive Data-Driven Policy Learning]\label{prob:learning}
    For the unknown system matrices $(A,B)$, use the data $D_t$ to construct a \emph{data-driven SDP}. Further, design a recursive update mechanism $\mathcal{P}$ such that the gain $K_t$ in \eqref{eq:learning-mechanism} tracks the optimal gain $K^*$ as new closed-loop data become available.
\end{problem}
}

In the rest of this section, we provide two well-known data-driven parameterization in the literature, namely, \emph{direct} and \emph{indirect} data-driven control, to construct the data-driven SDP. 

\subsection{Indirect Data-Driven SDP}\label{sec:indirect}

The first method is termed as \emph{indirect} data-driven. The key idea is two steps: i) use data $D_t$ to find an estimate $(\hat A_t, \hat B_t)$ of the unknown matrices $(A,B)$; ii) formulate the SDP with $(\hat A_t,\hat B_t)$. Based on the subspace relation \eqref{eq:data-relation} and disregard the noise measures $W_{0,t}$, $(\hat A_t,\hat B_t)$ can be estimated by solving the following ordinary least square problem 
\begin{equation}\label{eq:ols}
    (\hat A_t,\hat B_t)=\mathop{\arg\min}_{A,B}~\|\overline{X}_{1,t}-A \overline{X}_{0,t}-B\overline{U}_{0,t}\|_F^2
\end{equation}
where we have adopted the sample covariance matrices:
\begin{equation}\label{eq:covariance}
\begin{aligned}
\overline X_{1,t}&:=X_{1,t}\begin{bmatrix}
    U_{0,t}\\X_{0,t}
\end{bmatrix}^\top /t\quad     &&\overline X_{0,t}:=X_{0,t}\begin{bmatrix}
    U_{0,t}\\X_{0,t}
\end{bmatrix}^\top /t\\
\overline U_{0,t}&:=U_{0,t}\begin{bmatrix}
    U_{0,t}\\X_{0,t}
\end{bmatrix}^\top /t\quad 
    &&\overline W_{0,t}:= W_{0,t}\begin{bmatrix}
    U_{0,t}\\X_{0,t}
\end{bmatrix}^\top/t\\
\end{aligned}
\end{equation}

When the data is \emph{persistently exciting} \cite{willems2005note}, i.e., $\mathrm{rank}\left(\begin{bmatrix}
    U_{0,t}\\X_{0,t}
\end{bmatrix}\right)=m+n$, the solution of \eqref{eq:ols} is unique. {Taking the LQR problem \eqref{eq:model-lqr} as an example, the indirect data-driven SDP is given by 
\begin{equation}\label{eq:indirect-lqr}
    \begin{aligned}
        \mathop{\min}_{\Sigma,Y,L}~&\mathrm{Tr}(Q\Sigma)+\mathrm{Tr}(RL)\\
        \mathrm{subject~to}~&\begin{bmatrix}
            \Sigma-I_n&\hat A_t\Sigma+\hat B_tY\\\star &\Sigma
        \end{bmatrix}\succeq 0\\
        &\begin{bmatrix}
            L&Y\\Y^\top& \Sigma
        \end{bmatrix}\succeq 0\\
        &(\hat A_t,\hat B_t)=\mathop{\arg\min}_{A,B}~\|\overline{X}_{1,t}-A \overline{X}_{0,t}-B\overline{U}_{0,t}\|_F^2
    \end{aligned}
\end{equation}
This is a bilevel problem, but the inner one can be explicitly solved as $[\hat A_t\quad \hat B_t]=\overline{X}_{1,t}\begin{bmatrix}
    \overline{X}_{0,t}\\\overline{U}_{0,t}
\end{bmatrix}^\dag$. From this expression, it can also be seen that data enters the problem nonlinearly.
}

\subsection{Direct Data-Driven SDP}\label{sec:direct}

One notable feature of the linear policy optimization problems is that the system matrices $(A,B)$ always appear through the closed-loop matrix $A+BK$ in the optimization problem, for example, \eqref{eq:po}. Recent advances \cite{de2019formulas,zhao2025data} show that $A+BK$ can be \emph{directly} expressed using the data $D_t$ by parameterizing the control policy as
\begin{equation}\label{eq:policy}
    \begin{bmatrix}
        K_t\\I_n
    \end{bmatrix}=\Phi_t V_t,
\end{equation}
where $V_t \in \mathbb{R}^{(n+m)\times n}$ is a new decision variable. Here, $\Phi_t$ is constructed from the data matrices and is defined by
\begin{equation}
\Phi_t:=\begin{bmatrix}
    U_{0,t}\\X_{0,t}
\end{bmatrix}\begin{bmatrix}
    U_{0,t}\\X_{0,t}
\end{bmatrix}^\top/t.
\end{equation}
The matrix $\Phi_t$ is full rank if the data is persistently exciting. Using the parameterization \eqref{eq:policy} and the input-state relation \eqref{eq:data-relation}, the closed-loop matrix can be written as
\begin{equation}
    A+BK_t=[B\quad A]\begin{bmatrix}
        K_t\\I_n
    \end{bmatrix}=[B\quad A]\Phi_tV_t=(\overline{X}_{1,t}-\overline{W}_{0,t})V_t.
\end{equation}

By disregarding the unknown disturbance term $\overline{W}_{0,t}$ in the parameterized closed-loop matrix, we obtain a data-driven parameterization of $A + B K_t=\overline{X}_{1,t}V_t$. This parameterization can then be used to formulate the data-driven SDP by replacing $(A,B)$ with data $D_t$. {For example, the direct data-driven SDP for the LQR problem \eqref{eq:po} is given by
\begin{equation}\label{eq:dee-lqr}
    \begin{aligned}
        (X_t,Y_t)=\mathop{\arg\min}_{X,Y}~&\mathrm{trace}(Q\overline{X}_{0,t}Y)+\mathrm{trace}(X)\\
        \mathrm{subject~to}~&
        \begin{bmatrix}
            \overline{X}_{0,t}Y-I_n&\overline{X}_{1,t}Y\\
            \star&\overline{X}_{0,t}Y
        \end{bmatrix}\succeq 0,\\
        &
        \begin{bmatrix}
            X&\sqrt{R}\overline{U}_{0,t}Y\\
            \star & \overline{X}_{0,t}Y
        \end{bmatrix}\succeq 0.
    \end{aligned}
\end{equation}
}

In this program, the original variable $V$ is replaced by $X$ and $Y$ to ensure convexity. Different from the indirect SDP, data enters linearly in the direct SDP. The optimal gain is given by $K_t = \overline{U}_{0,t} Y_t (\overline{X}_{0,t} Y_t)^{-1}$, and the controllability Gramian is given by $\Sigma_t = \overline{X}_{0,t} Y_t$. It is known that \eqref{eq:dee-lqr} is equivalent to the corresponding indirect data-driven problem.

\begin{remark}
In \cite{de2019formulas} and \cite{dorfler2023certainty}, a simple formulation is provided based on the data matrices \eqref{eq:data} rather than the covariance matrices \eqref{eq:covariance}. Compared with this formulation, \eqref{eq:dee-lqr} utilizes empirical covariance matrices and offers several advantages, particularly for SDP formulations. First, the dimension of the decision variable $V_t$ does not increase with time. This is a notable feature for SDPs, as the dimension of the constraints remains fixed, in contrast to the formulations in \cite{de2019formulas,dorfler2023certainty}. Second, the covariance parameterization admits no nullspace in \eqref{eq:policy} and thus it achieves equivalence to the indirect certainty-equivalence problem without requiring explicit regularization. For further discussion, readers are referred to \cite[Section III]{zhao2025data}. 
\end{remark}

For both direct and indirect cases, we propose the data-driven SDP for the policy optimization problem takes the form
\begin{equation}\label{eq:opt-sdp}\tag{SDP-PO}
\begin{aligned}
    z^*=\mathop{\arg\min}_{z\in \mathbb{R}^N}~&c^\top z\\
    \mathrm{subject~to~}&F_i(D_t;z)\succeq 0,\quad i\in \{1,\ldots,l\}
\end{aligned}
\end{equation}
where $c\in \mathbb{R}^N$ is a known vector, and every $F_i(D_t;z)\in \mathbb{R}^{q_i\times q_i}$ is a squared matrix that parameterized by data $D_t$, and is linear in the decision variables $z$. The optimal control gain can be recovered from $z^*$ by a transformation $K^*=\mathcal{M}(z^*)$, where $\mathcal{M}:\mathbb{R}^N \to \mathbb{R}^{m\times n}$.


\section{Data-Enabled Primal-Dual Algorithm}\label{sec:algorithm}
In this section, we propose an policy learning mechanism $\mathcal{P}$ to solve Problem \ref{prob:learning}, which is summarized in Algorithm \ref{al:online-control-first-order}.

\begin{algorithm}
\caption{Data-Enabled Primal-Dual Algorithm}
\label{al:online-control-first-order}
\begin{algorithmic}[1]
\State Sample offline data $D_{t_0}$. Initialize the decision variables $z_{t_0}$, $S_{i,t_0}$, and $\Lambda_{i,t_0}$, $i\in\{1,\ldots,l\}$, by solving the offline data-driven SDP. Recover the offline control policy $K_{t_0}$ from $z_{t_0}$.
    \For{$t=t_0+1,\ldots$}
    \State $z$-update: solve
    \[
    z_{t}=\mathop{\arg\min}_{z}~c^\top z+\sum_{i=1}^l\frac{\rho}{2}\left\|S_{i,t-1}-F_i(D_{t};z)+\frac{1}{\rho}\Lambda_{i,t-1}\right\|_F^2.
    \]
    \State $S_i$-update:
    \[
    S_{i,t}=\Pi_{\mathbb{S}_{+}}\left( F_i(D_{t};z_{t})-\frac{1}{\rho}\Lambda_{i,t-1}\right), \quad i\in\{1,\ldots,l\}.
    \]
    \State $\Lambda_i$-update:
    \[
    \Lambda_{i,t}=\Lambda_{i,t-1}+\rho\left(F_i(D_{t};z_{t})-S_{i,t}\right), \quad i=\{1,\ldots,l\}.
    \]
    \State Policy update: recover $K_{t}$ from $z_{t}$ using $D_t$.
    \State Apply the control input $u_{t}=K_{t}x_{t}+e_{t}$, sample the new state $x_{t+1}$ from the closed-loop system
    \[
    x_{t+1}=Ax_t+BK_tx_t+Be_t+w_t,
    \]
    and update $D_{t+1}$ using $D_t$, $x_{t+1}$ and $u_t$.
    \EndFor
\end{algorithmic}
\end{algorithm}

We start from the following reformulation of \eqref{eq:opt-sdp}:
\begin{equation}\label{eq:problem}
    \begin{aligned}
        \min_{z,S_i}~&c^\top z+\sum_{i=1}^l\mathbb{I}_{\mathbb{S}_+}(S_i)\\
        \mathrm{subject~to}~&S_i=F_i(D_t;z), \quad i\in\{1,\ldots,l\}.
    \end{aligned}
\end{equation}
Here, we lift the SDP constraints by introducing auxiliary variables $S_i$, $i\in\{1,\ldots,l\}$, and using the indicator function $\mathbb{I}_{\mathbb{S}_+}(\cdot)$, which equals zero for positive semidefinite matrices and positive infinity otherwise.

The scaled-form augmented Lagrangian for \eqref{eq:problem} is given by
\begin{equation}
    \mathcal{L}_\rho
    =
    c^\top z
    +
    \sum_{i=1}^l \mathbb{I}_{\mathbb{S}_+}(S_i)
    +
    \sum_{i=1}^l\frac{\rho}{2}\left\|S_i-F_i(D_t;z)+\frac{1}{\rho}\Lambda_i\right\|_F^2.
\end{equation}
\textbf{Steps 3--5} of Algorithm \ref{al:online-control-first-order} follow the operator splitting method \cite{wen2010alternating}. The decision variables are split into two groups: $z$ and $S_i$. The initial decision variables $z_{t_0}$, $S_{i,t_0}$, and $\Lambda_{i,t_0}$, $i\in\{1,\ldots,l\}$, are obtained by solving \eqref{eq:dee-sdp2} with offline data $D_{t_0}$.

In \textbf{Step 3}, the primal variable $z$ is updated by solving an unconstrained quadratic programming problem with fixed $S_{i,t-1}$ and $\Lambda_{i,t-1}$. This is equivalent to solving a linear system in $z$. This step can be very efficient in practice, especially when some of the matrices can be preprocessed.

In \textbf{Step 4}, the auxiliary variables $S_i$, $i\in\{1,\ldots,l\}$, are updated by projecting onto the positive semidefinite cone. This follows from minimizing the augmented Lagrangian:
\begin{equation}\label{eq:16}
\begin{aligned}
    S_{i,t}
    =
    \mathop{\arg\min}_{S_i}~  
    \mathbb{I}_{\mathbb{S}_+}(S_i)
    +
    \frac{\rho}{2}\left\|S_{i}-F_{i}(D_{t};z_{t})+\frac{1}{\rho}\Lambda_{i,t-1}\right\|_F^2.
\end{aligned}
\end{equation}
The optimization problem is decomposable with respect to each $S_i$. Each individual optimization problem  can be solved by an eigenvalue decomposition, and thresholding the eigenvalue decomposition of $F_i(D_{t};z_{t})-\frac{1}{\rho}\Lambda_{i,t-1}$ gives
\begin{equation*}
    F_i(D_{t};z_{t})-\frac{1}{\rho}\Lambda_{i,t}= V\Sigma V^\top.
\end{equation*}
Using the orthogonal invariance of $\|\cdot\|_F$, \eqref{eq:16} can be rewritten as
\begin{equation}\label{eq:svd-1}
\begin{aligned}
    S_{i,t}=\mathop{\arg\min}_{S_i}~&
    \left\|V^\top S_{i}V-\Sigma\right\|_F^2\\
    \mathrm{subject~to}~&S_i\succeq 0.
\end{aligned}
\end{equation}
Since $V^\top S_i V \succeq 0$ if and only if $S_i \succeq 0$, the optimal solution of \eqref{eq:svd-1} is given by
\begin{equation}
    V^\top S_i V=\Pi_{\mathbb{S}_+}(\Sigma),
\end{equation}

In \textbf{Step 5}, the dual variables $\Lambda_{i,t}$ are updated by gradient ascent, where $\rho$ is the augmented Lagrangian parameter.

In \textbf{Step 6}, the control policy $K_{t}$ is recovered from $z_{t}$.

In \textbf{Step 7}, the real system applies the control input $u_t=K_tx_t+e_t$, where $e_t\in\mathbb{R}^m$ is a probing noise that encourages exploration and ensures persistency of excitation of the data matrices.

We are interested in analyzing: i) the convergence of $z_t$, which reflects the convergence of the control gain $K_t$; and ii) the constraint satisfaction
$F_i([X_{0,t}\quad X_{1,t}-W_{0,t}\quad U_{0,t}];z_{t+1})\succeq 0, \quad i\in\{1,\ldots,l\}$
with respect to the noiseless data.

\section{Convergence Analysis}\label{sec:analysis}

For analysis purposes, we further reformulate \eqref{eq:opt-sdp} into the following standard form:
\begin{equation}\label{eq:dee-sdp2}
    \begin{aligned}
        \min_{z,s}~&f_1(z)+f_2(s)\\
        \mathrm{subject~to}~&\underbrace{A_1(D_t)}_{:=A_{1,t}}z+A_2s=b.
    \end{aligned}
\end{equation}
In the above formulation,
$s:=
[\mathrm{vec}(S_1)^\top\quad 
\ldots\quad
\mathrm{vec}(S_l)^\top]^\top
\in\mathbb{R}^q$, $f_1(z)$ corresponds to the cost term $c^\top z$, and $f_2(s)$ corresponds to the indicator function term $\sum_{i=1}^l \mathbb{I}_{\mathbb{S}_+}(S_i)$. The associated dual variable is $\lambda:=
\begin{bmatrix}
\mathrm{vec}(\Lambda_1)^\top&
\mathrm{vec}(\Lambda_2)^\top&
\ldots&
\mathrm{vec}(\Lambda_l)^\top
\end{bmatrix}^\top.$
The matrices
\begin{equation}\label{eq:noiseless-map}
\begin{aligned}
    A_{1,t}
    &:=
    A_1([{X}_{0,t}\quad {U}_{0,t}\quad {X}_{1,t}])
    \in \mathbb{R}^{q\times N},\\
    A_{1,t}^{\mathrm{nl}}
    &:=
    A_1([{X}_{0,t}\quad {U}_{0,t} \quad {X}_{1,t}-{W}_{0,t}])
    \in \mathbb{R}^{q\times N}
\end{aligned}
\end{equation}
denote the data-induced affine constraint maps parameterized by the raw data and the noiseless data, respectively.

From the SDP structure in \eqref{eq:opt-sdp}, we have $A_2=-I$, with dimension consistent with the constraint matrices $F_i(D_t;z)$, and $b$ is a constant vector arising from the constant terms in $F_i(D_t;z)$. It should be noted that $A_2$ and $b$ do not vary with time due to the structure of the adaptive control problems considered here: the data matrices always multiply the decision variable $z$. Nevertheless, the analysis in this section can be extended to the case where $A_2$ and $b$ are also time-varying.

{For notational convenience, denote
\begin{equation}\label{eq:kkt}
    \omega^*_t:=
    \begin{bmatrix}
        z_t^*\\s_t^*\\\lambda_t^*
    \end{bmatrix}
    :=
    \begin{bmatrix}
        z^*(A_{1,t})\\
        s^*(A_{1,t})\\
        \lambda^*(A_{1,t})
    \end{bmatrix}
\end{equation}
as the KKT point of \eqref{eq:dee-sdp2} associated with $A_{1,t}$ to which the iterates of Algorithm~\ref{al:online-control-first-order} converge when $A_{1,t}$ is held fixed. The existence of such a KKT point is ensured by \cite{wen2010alternating}.} We first make the following key assumptions and then investigate several properties of the SDP problem \eqref{eq:dee-sdp2}.

The first assumption concerns the linear structure of $A_{1,t}$.
\begin{assumption}[Linearity in Data]\label{ass:linear}
    For the direct case in Section \ref{sec:direct}, there exists a bounded linear operator $\mathcal{L}^d$ such that $A_{1,t}=\mathcal{L}^d(\overline{X}_{0,t}, \overline{X}_{1,t},\overline{U}_{0,t})$. For the indirect case, there exists a bounded linear operator $\mathcal{L}^i$ such that $A_{1,t}=\mathcal{L}^i(\hat A_t,\hat B_t)$. \hfill 
\end{assumption}
The two cases are discussed separately because $(\hat A_t,\hat B_t)$ is not linear in the data $D_t$. The two LQR formulations \eqref{eq:indirect-lqr} and \eqref{eq:dee-lqr} nicely illustrate Assumption \ref{ass:linear}.

The following assumption concerns the rank of $A_{1,t}$.
\begin{assumption}[Uniform Rank Nondegeneracy]\label{ass:pe}
    For any $t\ge t_0$, the matrix $A_{1,t}$ has full column rank in the sense that $\underline{\sigma}(A_{1,t})\ge \underline{\gamma}$ for some $\underline{\gamma}>0$, where $\underline{\sigma}(\cdot)$ denotes the smallest singular value of a matrix. \hfill 
\end{assumption}

For our problem \eqref{eq:opt-sdp}, Assumption \ref{ass:pe} captures the persistent excitation of data, in the sense that the sample covariance matrices $\overline{X}_{0,t}$, $\overline{X}_{1,t}$, and $\overline{U}_{0,t}$ have sufficient rank, with positive smallest nonzero singular values. This assumption is ensured with probability one under the control input $u_t=K_tx_t+e_t$ if the probing noise $e_t$ is properly designed \cite{faradonbeh2020input}. In the following, we take the LQR problem as an example to show that both Assumptions \ref{ass:linear} and \ref{ass:pe} can be ensured.

\begin{example}[LQR Control]\label{rem:lqr}
    Consider the LQR problem as an example. For the direct case \eqref{eq:dee-lqr}, the matrices
    \[
    \begin{bmatrix}
    \overline{X}_{0,t}Y-I_n&\overline{X}_{1,t}Y\\
    \star&\overline{X}_{0,t}Y
    \end{bmatrix}
    \quad \text{and} \quad
    \begin{bmatrix}
        X&\sqrt{R}\overline{U}_{0,t}Y\\
        \star & \overline{X}_{0,t}Y
    \end{bmatrix}
    \]
    can be stacked and expressed in the form
\begin{equation*}
    \underbrace{\begin{bmatrix}
        0_{n^2\times n^2}&I_n \otimes \overline{X}_{0,t}\\
        0_{n^2\times n^2}&I_n\otimes \overline{X}_{1,t}\\
        0_{n^2\times n^2}&I_n\otimes \overline{X}_{0,t}\\
        I_{n^2}&0_{n^2\times n^2}\\
        0_{n^2\times n^2}&I_n\otimes\sqrt{R}\overline{U}_{0,t}\\
        0_{n^2\times n^2}&I_n\otimes \overline{X}_{0,t}
    \end{bmatrix}}_{A_{1,t}}
    \underbrace{\begin{bmatrix}
        \mathrm{vec}(X)\\
        \mathrm{vec}(Y)
    \end{bmatrix}}_{z}.
\end{equation*}
The matrix $A_{1,t}$ is clearly linear in the data matrices, which shows that Assumption \ref{ass:linear} holds. We can see that $\underline{\sigma}(A_{1,t})>0$ if and only if $[
    3\overline{X}_{0,t}^\top \quad  \overline{X}_{1,t}^\top \quad \sqrt{R}\overline{U}_{0,t}^\top ]^\top$
has full column rank. Since $R\succ 0$, the latter is equivalent to
\begin{equation}\label{eq:null-space}
    \mathcal{N}(\overline{X}_{0,t})
    \cap
    \mathcal{N}(\overline{X}_{1,t})
    \cap
    \mathcal{N}(\overline{U}_{0,t})
    =
    \{0\}.
\end{equation}
Relation \eqref{eq:null-space} holds under the persistency-of-excitation condition in Assumption \ref{ass:pe}, which implies that $\Phi_{t}=D_{0,t}D_{0,t}^\top/t$ has full rank. \hfill
\end{example}

The following assumption concerns regularity of \eqref{eq:dee-sdp2}.

\begin{assumption}[Strong Regularity]\label{ass:regularity}
There exists a nonempty compact set $\mathcal{A}\subset \mathbb{R}^{q\times N}$ such that
\begin{equation} A_{1,t},A_{1,t}^{\mathrm{nl}} \in \mathcal{A},\quad \forall t\ge t_0.
\end{equation}
Moreover, for every $\tilde A_1\in \mathcal{A}$, strict complementarity and primal-dual nondegeneracy hold at the KKT point $[z^\star(\tilde A_1),s^*(\tilde A_1),\lambda^*(\tilde A_1)].$ that defined in \eqref{eq:kkt}.
\end{assumption}

The compactness condition in Assumption 3 should be interpreted as a deterministic regularity condition on the realized sequence of data-dependent SDPs. As shown later in Lemmas~\ref{lem:dos-evolution} and~\ref{lem:srg-evolution}, under the stochastic setting considered therein, stable closed-loop trajectories and Gaussian noises imply that the data-dependent matrices remain bounded in probability, or equivalently, belong to a compact set on any finite horizon with high probability. Also, it has been shown that primal and dual nondegeneracy, as well as strict complementarity, are generic properties for SDPs \cite{alizadeh1997complementarity}. 

Under Assumption \ref{ass:regularity}, the Jacobian matrix of the KKT system is nonsingular at the KKT point \eqref{eq:kkt}. This implies that the KKT solution map is locally single-valued and Lipschitz continuous with respect to $A_{1,t}$ \cite[Theorem 18]{chan2008constraint}, \cite[Theorem 3.1]{alizadeh1998primal}. The result is formalized in the following lemma.

\begin{lemma}[Local Smoothness]\label{lem:smoothness}
Consider the problem \eqref{eq:dee-sdp2}, and let Assumptions \ref{ass:linear}, \ref{ass:pe}, and \ref{ass:regularity} hold. Then, there exists a constant $L > 0$ such that for any $t_1,t_2\ge t_0$:
\[
\|\omega^\star(A_{1,t_1}) - \omega^\star(A_{1,t_2})\|_F
\le
L \|A_{1,t_1} - A_{1,t_2}\|_F.
\]
\end{lemma}

We introduce two quantities to characterize the online data-driven SDP \eqref{eq:dee-sdp2}. The first measures how much the SDP changes when a new data sample $(x_t,u_t)$ arrives, thereby quantifying the temporal drift of the constraints. The second measures how strongly the informative data component dominates the disturbance-induced perturbation, thereby quantifying the reliability of the data-dependent formulation. These notions motivate the definitions of Difference-of-Signal (DOS) and Sim-to-Real Gap (SRG).

\begin{definition}[Difference-of-Signal (DOS)]
    For problem \eqref{eq:dee-sdp2}, the Difference-of-Signal (DOS) at time $t$ is defined as
    \begin{equation}\label{eq:dos}
        \mathrm{DOS}_t
    :=
    \bigl\| A_{1,t}-A_{1,t-1}\bigr\|_F
    =
    \bigl\| A_1(D_{t})-A_1(D_{t-1})\bigr\|_F.
    \end{equation}
\end{definition}

DOS is analogous to the variation measures used in online optimization and learning \cite{dixit2019online,hall2015online,jadbabaie2015online}, where dynamic regret bounds are often expressed in terms of \emph{temporal variation}, \emph{path length}, or \emph{variation budgets} of the loss functions, comparators, or target sequences. Unlike these optimizer-level quantities, DOS measures the one-step drift of the data-induced affine constraint map $A_{1,t}$, which in turn indirectly measures the drift of the optimizer according to Lemma \ref{lem:smoothness}.

\begin{definition}[Sim-to-Real Gap (SRG)]
    For problem \eqref{eq:dee-sdp2}, the Sim-to-Real Gap (SRG) at time $t$ is defined as
    \begin{equation}
        \mathrm{SRG}_t:=\|A_{1,t}-A_{1,t}^{\mathrm{nl}}\|_F.
    \end{equation}
\end{definition}

SRG quantifies the mismatch between the ideal data-dependent problem and its noisy finite-sample realization, similar to noise qualification in data-driven control. For indirect formulations in Section \ref{sec:indirect}, where one first identifies a model and then solves a model-based control problem, SRG reduces to the usual identification error induced by data noise. Hence, for the considered OLS identification \eqref{eq:ols} it is governed by the inverse of the classical signal-to-noise ratio \cite{dorfler2023certainty}, where $\mathrm{SNR}_t:=\frac{\underline{\sigma}({\Phi_t})}{\underline{\sigma}(\overline{W}_{0,t})}$.



{To characterize the asymptotic behavior of $\mathrm{DOS}_t$ and $\mathrm{SRG}_t$, we specialize the analysis to the following stochastic setting. Assume that $\{w_t\}_{t\ge0}$ and $\{e_t\}_{t\ge0}$ are mutually independent Gaussian white-noise sequences with
$w_t\sim\mathcal N(0,I_n),
e_t\sim\mathcal N(0,I_m).$
Throughout this subsection, $O_p(\cdot)$ denotes the usual order in probability with respect to the probability law induced by these stochastic processes.}

\begin{lemma}[Evolution of DOS]
\label{lem:dos-evolution}
Suppose that the closed-loop trajectory is stable in the sense that of finite variance
\begin{equation}\label{eq:second-moment}
    \sup_{t\ge0}\mathbb E\|x_t\|^2<\infty.
\end{equation}
Under Assumptions~\ref{ass:linear} and \ref{ass:pe}, there exist constants $L_{\mathrm{DOS}}^{\mathrm d}>0$ and $L_{\mathrm{DOS}}^{\mathrm i}>0$ such that the DOS of the direct data-driven formulation \eqref{eq:opt-sdp} satisfies
\begin{equation}\label{eq:direct-dos}
    \mathrm{DOS}^{\mathrm d}_{t}
    \le
    L_{\mathrm{DOS}}^{\mathrm d}
    \left\|
    \begin{bmatrix}
        \overline X_{1,t}-\overline X_{1,t-1} \\
        \overline X_{0,t}-\overline X_{0,t-1} \\
        \overline U_{0,t}-\overline U_{0,t-1}
    \end{bmatrix}
    \right\|_F
    =
    \mathcal{O}_p\!\left(\frac{1}{t}\right),
\end{equation}
whereas the DOS of the indirect formulation satisfies
\begin{equation}\label{eq:indirect-dos}
    \mathrm{DOS}^{\mathrm i}_{t}
    \le
    L_{\mathrm{DOS}}^{\mathrm i}
    \left\|
        \overline X_{1,t}\Phi_{t}^{-1}
        -
        \overline X_{1,t-1}\Phi_{t-1}^{-1}
    \right\|_F
    =
    \mathcal{O}_p\!\left(\frac{1}{t}\right).
\end{equation}
\end{lemma}

\begin{lemma}[Evolution of SRG]
\label{lem:srg-evolution}
Suppose that the closed-loop trajectory is stable in the sense of \eqref{eq:second-moment}. Under Assumptions~\ref{ass:linear} and \ref{ass:pe}, there exist constants $L_{\mathrm{SRG}}^{\mathrm d}>0$ and $L_{\mathrm{SRG}}^{\mathrm i}>0$ such that the SRG of the direct data-driven formulation \eqref{eq:opt-sdp} satisfies
\begin{equation}\label{eq:direct-srg}
    \mathrm{SRG}^{\mathrm d}_t
    \le
    L_{\mathrm{SRG}}^{\mathrm d}
    \left\|
        \overline W_{0,t}
    \right\|_F
    =
    \mathcal{O}_p\!\left(\frac{1}{\sqrt t}\right),
\end{equation}
whereas the SRG of the indirect formulation satisfies
\begin{equation}\label{eq:indirect-srg}
    \mathrm{SRG}^{\mathrm i}_t
    \le
    L_{\mathrm{SRG}}^{\mathrm i}
    \left\|
        \overline W_{0,t}\Phi_t^{-1}
    \right\|_F
    =
    \mathcal{O}_p\!\left(\frac{1}{\sqrt t}\right).
\end{equation}
\end{lemma}

The bounded second-moment condition in \eqref{eq:second-moment} is standard in online data-driven policy learning, such as \cite{zhao2025data}. The notation $\mathcal{O}_p(1/\sqrt t)$ should be understood in the usual probabilistic sense. For example, $\mathrm{SRG}_t=\mathcal{O}_p(1/\sqrt t)$ means that $\sqrt t\,\mathrm{SRG}_t$ is bounded in probability; equivalently, for any $\delta\in(0,1)$, there exist constants $M_\delta>0$ and $t_\delta>0$ such that $\mathbb P(\mathrm{SRG}_t\le M_\delta/\sqrt t)\ge 1-\delta$ for all $t\ge t_\delta$. Therefore, Lemma~\ref{lem:srg-evolution} does not imply that the SRG decreases monotonically along every sample path. Rather, it shows that the noise-induced mismatch between the noisy and noiseless data-dependent SDPs vanishes at the standard Monte Carlo rate in probability.

From these lemmas, one can see that the main difference between direct and indirect approaches is that the indirect formulation always introduces the regressor $\Phi_t^{-1}$. Hence, under diminishing input excitation, the direct quantities may decrease with the covariance perturbations, whereas the indirect ones can be amplified by the growth or ill-conditioning of $\Phi_t^{-1}$.

We are now ready to analyze the convergence of Algorithm \ref{al:online-control-first-order} for \eqref{eq:dee-sdp2}, where we will show that both DOS and SRG influence the result.

\subsection{Local Convergence}\label{sec:local}

We begin with a local convergence analysis under the conditions that: i) $\omega_{t_0}$ is obtained by warm-starting with offline data $D_{t_0}$; and ii) $A_{1,t}$ varies sufficiently slowly, namely, $\mathrm{DOS}_t$ is uniformly bounded. {Define
\begin{equation}
    f^*:=c^\top z^*(A_{1,t}^{\mathrm{nl}})
\end{equation}
be the optimal cost of \eqref{eq:opt-sdp} with noiseless data. Under the assumption of persistently exciting data, the noiseless data uniquely determine the ground-truth model $(A,B)$. Consequently, $f^*$ coincides with the optimal cost of the ground-truth model-based problem and is therefore time-invariant. We note that, in the direct case, $z^*(A_{1,t}^{\mathrm{nl}})$ and $s^*(A_{1,t}^{\mathrm{nl}})$ depend on $t$, since the data matrices enter the SDP directly and thus influence the optimal solution. In contrast, in the indirect case, any noiseless data satisfying Assumption~\ref{ass:pe} yield the same SDP. As a result, $z^*(A_{1,t}^{\mathrm{nl}})$ and $s^*(A_{1,t}^{\mathrm{nl}})$ are independent of $t$.}


The convergence result is provided in the following theorem.
\begin{theorem}\label{th:opt-gap-linear}
Consider \eqref{eq:dee-sdp2}, and let Assumptions \ref{ass:linear}--\ref{ass:regularity} hold. Then, for any $T\ge 0$, there exist constants $\varepsilon^\omega_{T}>0$, $\varepsilon^{\mathrm{DOS}}_{T}>0$, $\sigma_T\in(0,1)$ and $0<C_1,C_2,C_3<\infty$, such that, if
\[
\omega_{t_0}\in \mathcal{B}(\omega_{t_0}^*,\varepsilon^\omega_T)\quad  \mathrm{and} \quad
\mathrm{DOS}_t \le \varepsilon_T^{\mathrm{DOS}}
\]
for all $t\in[t_0,\ldots,t_0+T-1]$, then it holds that
\begin{equation}\label{eq:opt-gap-linear}
\begin{aligned}
   &\left|f_1(z_{t+1})-f^*\right|\\
\le{}&C_1\sigma_{T}^{t+1-t_0}
+
C_2\sum_{k=t_0}^{t+1}\sigma_{T}^{t-k}\mathrm{DOS}_k
+
C_3\mathrm{SRG}_{t+1}.
\end{aligned}
\end{equation}
\end{theorem}



Theorem \ref{th:opt-gap-linear} deterministically characterizes the optimization tracking error in terms of $\mathrm{DOS}_t$ and $\mathrm{SRG}_t$. Under the additional stochastic assumptions introduced before Lemmas \ref{lem:dos-evolution} and \ref{lem:srg-evolution}, these quantities admit explicit asymptotic rates. Looking at the result \eqref{eq:opt-gap-linear}, it shows that, under suitable regularity conditions, the optimality gap is bounded by three terms. The first term $C_1\sigma_T^{t+1-t_0}$ decays linearly and reflects the initialization error at $\omega_{t_0}$. The second term $C_2\sum_{k=t_0}^{t+1}\sigma_T^{t-k}\mathrm{DOS}_k$ captures the accumulated effect of problem variation. For constant DOS, this term remains bounded, while for vanishing $\mathrm{DOS}_t$, it converges to zero. The third term $C_3\mathrm{SRG}_{t+1}$ quantifies the mismatch between the data-driven problem and the ground-truth model-based problem. {Unlike the DOS term, which enters the recursive tracking dynamics and therefore accumulates over time, the SRG term is introduced only in the final comparison between the current current noisy SDP \eqref{eq:opt-sdp} and its noiseless counterpart. Hence, it contributes only through the instantaneous quantity $\mathrm{SRG}_{t+1}$.} 

Under the additional stochastic assumptions over $w_t$ and $e_t$, Lemma \ref{lem:srg-evolution} further implies that $C_3\mathrm{SRG}_{t}=O_p(1/\sqrt{t})$, which matches the order appearing in sample-based policy optimization and model-free LQR analyses \cite{fazel2018global,malik2020derivative,ju2025model}. In contrast to these LQR-focused results, our framework applies to a broader class of control problems that admit SDP formulations.

In practice, $\omega_{t_0}$ can be initialized within a neighborhood of $\omega_{t_0}^*$ by solving \eqref{eq:opt-sdp} using the offline data set $\mathcal D_{t_0}$. Moreover, Lemma~\ref{lem:dos-evolution} shows that, for stable closed-loop trajectories, DOS decreases over time. Therefore, with sufficient offline data, the subsequent online variation of the data-induced SDP can be made small. We note that the closed-loop stability is not discussed here, but the result follows from standard standard sequential stability analysis similar to \cite{zhao2025policy} when $\omega_{t_0}$ varies slowly and the SDP constraints are satisfied.  


\subsection{Global Convergence}

We next consider the case where the initialization $\omega_{t_0}$ is arbitrary and $\mathrm{DOS}_t$ is not assumed to be sufficiently small. This corresponds to the situation where the offline data are insufficient and the initialization is arbitrary. This case relies on significantly weaker assumptions than the local convergence analysis. Nevertheless, we can still certify the regret.


%


\begin{theorem}\label{th:opt-gap}
   Consider problem \eqref{eq:dee-sdp2}, and let Assumptions \ref{ass:linear}--\ref{ass:regularity} hold. Then, there exist constants $B_1,B_2,B_3>0$ such that
    \begin{equation}\label{eq:opt-gap}
        \frac{1}{T-t_0+1}
        \left|
        \sum_{t=t_0}^T
        \left(
        f_1(z_t)-f^*
        \right)
        \right|
        \le S_T,
    \end{equation}
    where
    \begin{equation}
        S_T
        =
        \frac{
        B_1+B_2\sum_{t=t_0}^T\sum_{k=t_0}^t \mathrm{DOS}_k
        }{T-t_0+1}
        +
        B_3\mathrm{SRG}_T.
    \end{equation}
\end{theorem}

The comparison between Theorems~\ref{th:opt-gap} and~\ref{th:opt-gap-linear} highlights two different convergence regimes. Theorem~\ref{th:opt-gap-linear} bounds the pointwise optimality gap, whereas Theorem~\ref{th:opt-gap} bounds regret. This is because, away from the KKT point, the primal-dual sequence may oscillate, making the pointwise optimality gap less regular. In contrast, regret corresponds to an ergodic objective measure in our setting, since the cost function is linear, and therefore provides a smoother convergence metric.

Another distinction is that Theorem~\ref{th:opt-gap-linear} establishes a local linear convergence result, with a term that scales linearly with $\mathrm{SRG}$ and another term that depends on the weighted cumulative variation of $\mathrm{DOS}_k$. In contrast, Theorem~\ref{th:opt-gap} gives a global sublinear result, again with a term that scales linearly with $\mathrm{SRG}$, but now involving a double accumulation of $\mathrm{DOS}_k$. This difference arises because the global contractivity property of the primal-dual sequence is weaker than the local contraction around the KKT point. The resulting sublinear rate is consistent with the classical convergence behavior of primal-dual methods \cite{he20121}.

\subsection{Constraint Tightening for Closed-Loop Guarantees}\label{sec:constraint-tightening}

Another problem of interest is constraint satisfaction with respect to the ground-truth model, namely, the positive definiteness of $F_i(D_t^{\mathrm{nl}};z_{t+1}), \quad
D_t^{\mathrm{nl}}:=
[\overline{X}_{0,t}\quad \overline{U}_{0,t}\quad \overline{X}_{1,t}-\overline{W}_{0,t}].$
Under the transformed form \eqref{eq:dee-sdp2}, this is equivalent to investigating the positive definiteness of $\mathrm{vec}^{-1}(A_{1,t}^{\mathrm{nl}}z_{t+1}-b)
=
\mathrm{diag}\left(
F_1(D_t^{\mathrm{nl}};z_{t+1}),\ldots,F_l(D_t^{\mathrm{nl}};z_{t+1})
\right).$ {Similar to policy optimization-based methods \cite{zhao2025policy}, we consider only the case in Section \ref{sec:local}, where $\omega_{t_0}$ is obtained by warm-starting and $\mathrm{DOS}_t$ is sufficiently small. The following proposition quantifies this constraint satisfaction property.

\begin{proposition}\label{prop:constraint}
    Consider problem \eqref{eq:dee-sdp2}, and let Assumptions \ref{ass:linear}--\ref{ass:regularity} hold. Then, for any $T\ge 0$, there exist constants $\varepsilon^\omega_{T}>0$, $\varepsilon^{\mathrm{DOS}}_{T}>0$, $\sigma_T\in (0,1)$ and $0<C_5,C_6,C_7<\infty$ such that, if
\[
\omega_{t_0}\in \mathcal{B}(\omega_{t_0}^*,\varepsilon^\omega_T)
\quad\text{and}\quad
\mathrm{DOS}_t\le \varepsilon^{\mathrm{DOS}}_{T},\quad \forall t\in[t_0,\ldots,t_0+T],
\]
then, for all $t\in[t_0,\ldots,t_0+T-1]$,
\begin{equation}\label{eq:50}
    \mathrm{vec}^{-1}(A_{1,t}^{\mathrm{nl}}z_{t+1}-b)+\bar{\epsilon}_t I \succeq 0,
\end{equation}
where
\begin{equation*}
    \bar{\epsilon}_t
    :=
    C_5 \sigma_T^{\,t-t_0}
    +
    C_6 \sum_{k=t_0}^{t+1}\sigma_T^{\,t-k}\mathrm{DOS}_k
    +
    C_7 \mathrm{SRG}_{t+1}
\end{equation*}
\end{proposition}}

Proposition \ref{prop:constraint} shows that the SDP constraints corresponding to the noiseless data are violated by at most $\bar{\epsilon}_t I$ at the $t$-th iterate. Moreover, this term vanishes if both $\mathrm{SRG}_t$ and $\mathrm{DOS}_t$ decrease to zero. In practice, to ensure anytime constraint satisfaction, one can tighten the constraints of the SDP problem \eqref{eq:opt-sdp} so that the tightening compensates for the worst-case violation $\bar{\epsilon}_t I$. 

Following this result, we provide a practical way to ensure constraint satisfaction via a tightening mechanism. 

\textbf{Constraint tightening}. Instead of solving problem \eqref{eq:dee-sdp2}, which corresponds to \eqref{eq:opt-sdp}, we use Algorithm \ref{al:online-control-first-order} to solve the following tightened problem:
\begin{equation}\label{eq:dee-sdp3}
    \min_{z,s}~f_1(z)+f_2(s)
    \quad
    \mathrm{subject~to}~
    A_{1,t} z+A_{2}s=b+\mathrm{vec}(\epsilon I_q).
\end{equation}
We note that the additive term $\mathrm{vec}(\epsilon I_q)$ does not change Assumptions \ref{ass:linear}, \ref{ass:pe}, and \ref{ass:regularity}. The following proposition show that the constraints for the original problem \eqref{eq:opt-sdp} are all satisfied all the time.

\begin{corollary}\label{coro:constraint-tightening}
    Consider problem \eqref{eq:dee-sdp3}, and let Assumptions \ref{ass:linear}-\ref{ass:regularity} holds. Then, for any $T\ge 0$, there exist constants $\varepsilon^\omega_{T}>0$, $\varepsilon^{\mathrm{DOS}}_{T}>0$, $\sigma_T\in (0,1)$ and $0<C_5,C_6,C_7<\infty$ such that, if\[
\omega_{t_0}\in \mathcal{B}(\omega_{t_0}^*,\varepsilon^\omega_T)
\quad\text{and}\quad
\mathrm{DOS}_t\le \varepsilon^{\mathrm{DOS}}_{T},\quad \forall t\in[t_0,\ldots,t_0+T],
\]
and $\epsilon>\bar \epsilon_t,\forall t\in[t_0,\ldots,t_0+T]$, where $\bar\epsilon_t$ is defined following \eqref{eq:50}.
Then, for all $t\in[t_0,\ldots,t_0+T-1]$,
\begin{equation}
    \mathrm{vec}^{-1}(A_{1,t}^{\mathrm{nl}}z_{t+1}-b)\succeq 0.
\end{equation}
\end{corollary}


    
    








\section{Simulations}
\label{sec:simulation}
In this section, we provide simulation results for LQR, $\mathcal{H}_\infty$, and safe exploration. All simulations are performed on a MacBook Air with an M3 chip, using MATLAB R2024b. The SDPs are solved by CVX with the default solver SDPT3 version 4.0. We first benchmark our method on the canonical LQR problem before moving to other objectives or state-constrained settings. 

\subsection{The LQR Problem}
\label{sec:lqr}

We take the benchmark system dynamics from \cite{dorfler2023certainty,zhao2025data}:
\begin{equation}\label{eq:baseline}
  A =
\begin{bmatrix}
1.01&0.01&0\\
0.01&1.01&0.01\\
0&0.01&1.01
\end{bmatrix},
\quad
B = I_3.
\end{equation}
Let $Q=I_3$ and $R=0.001\times I_3$ in the objective function. Similarly to \cite{zhao2025data}, we generate PE data $(X_{0,t_0},U_{0,t_0},W_{0,t_0})$ of length $t_0=20$ from the Gaussian distribution $\mathcal{N}(0,I_3)$, and $X_{1,t}$ is obtained from the system dynamics \eqref{eq:system}. The process noise is $w_t\sim\mathcal{N}(0,I_3)$, and the probing noise in the control input $u_t=K_tx_t+e_t$ is $e_t\sim \mathcal{N}(0,I_3)$.

Let $C(K)=\mathrm{Tr}[(Q+K^\top RK)\Sigma]$ be the LQR cost in \eqref{eq:po}, where $\Sigma$ is the solution of the Lyapunov equation $\Sigma = I_n+(A+BK)\Sigma (A+BK)^\top$, and let $C^*:=C(K^*)$ be the optimal cost that corresponds to the optimal LQR gain $K^*$ with respect to $(A,B)$. Define the regret as
\begin{equation*}
    \mathrm{Regret}_t:=\frac{1}{t}\sum_{k=t_0}^{t_0+t-1}(C(K_k)-C^*).
\end{equation*}
We compare our method with  state-of-the-art approach, DeePO \cite{zhao2025data}, which learns the optimal LQR gain directly from data using a policy optimization method \cite{fazel2018global}. For a fair comparison, both our method and DeePO \cite{zhao2025data} use the offline data $(X_{0,t_0},U_{0,t_0},W_{0,t_0})$ to solve the LQR problem and use the resulting solution to warm-start the online iterations. For DeePO to operate, we set $t_0=20$ so that the initial controller $K_{t_0}$ is stabilizing. In contrast, our method does not require an initially stabilizing controller, allowing the initial dataset to be substantially shorter. In our experiments, we found that our method requires only the minimum data length of $t_0=m+n=6$ to maintain closed-loop stability and learn the optimal gain, whereas DeePO requires $t_0=18$.

In Figure \ref{fig:lqr-reg}, the regrets of our method and DeePO are provided. It can be seen that both methods exhibit similar sublinear convergence rates. Our method shows relatively faster convergence, although the rate is of the same order of magnitude. Figure \ref{fig:lqr-opt} shows the optimality gap. During roughly the first 100 steps, our method shows faster convergence. After that, both methods show similar performance.

We have also calculated the computational time. In this case, both methods run for 2000 steps. DeePO takes $0.2919$s, while our method takes $0.6833$s. This is because our method needs to solve a linear equation and perform an SVD at every iteration, whereas DeePO only needs to solve a single Lyapunov equation and perform a simple subspace projection.

For fairness, we recall that, unlike DeePO tailored to LQR, our method applies to general SDP design formulations as shown in the next two sections.

\begin{figure}[!htbp]
    \centering
    \begin{subfigure}[t]{0.7\linewidth}
        \centering
        \includegraphics[width=\linewidth]{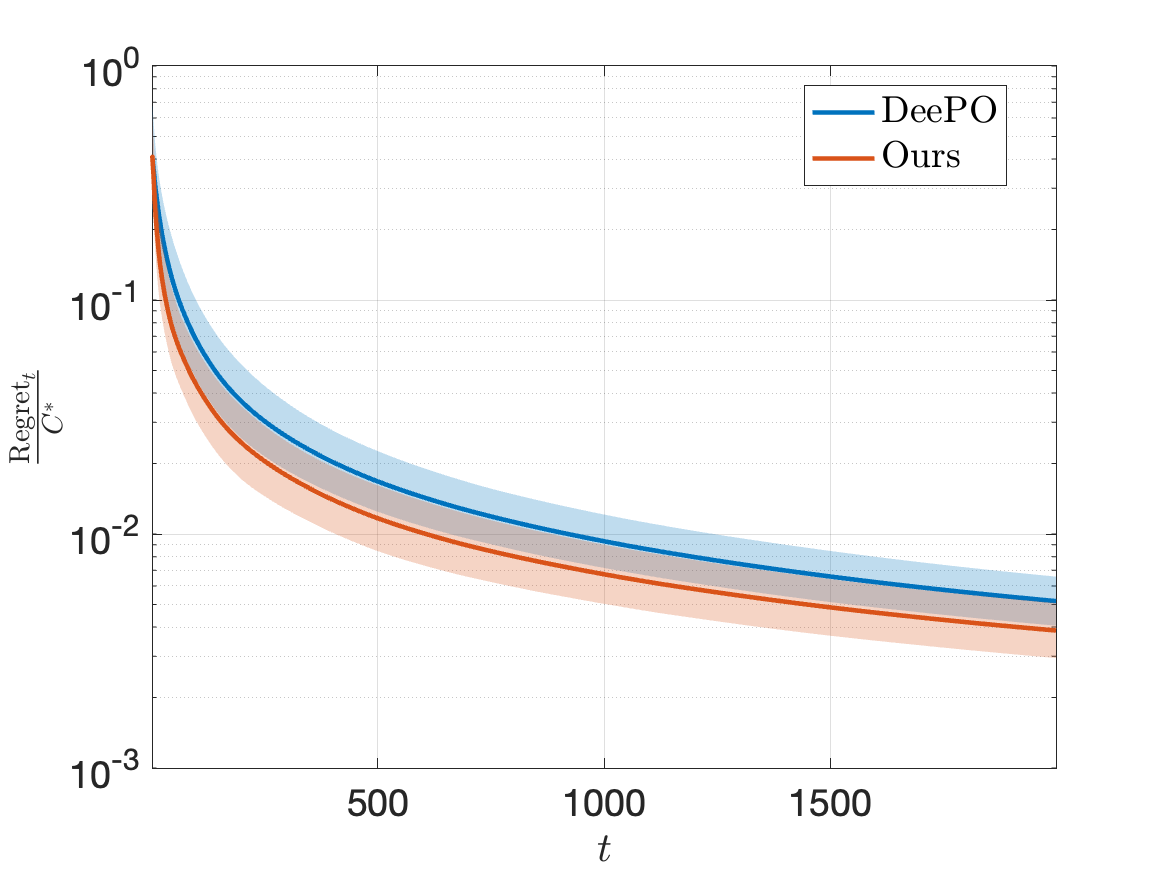}
        \caption{Regrets of our method and DeePO under the same setting.}
        \label{fig:lqr-reg}
    \end{subfigure}
    \hfill
    \begin{subfigure}[t]{0.7\linewidth}
        \centering
        \includegraphics[width=\linewidth]{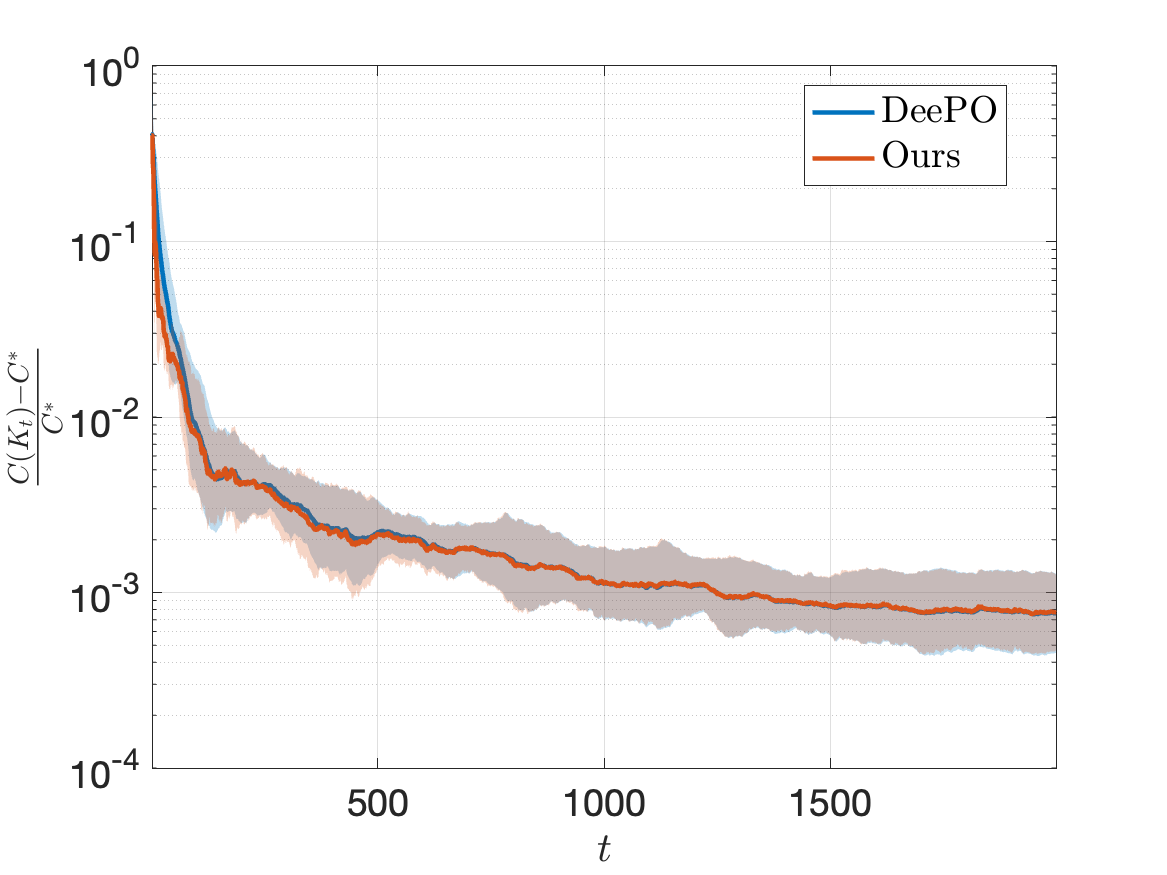}
        \caption{Optimality gaps of our method and DeePO under the same setting.}
        \label{fig:lqr-opt}
    \end{subfigure}
    \caption{Comparison between our method and DeePO for the LQR problem, where the shaded areas represent one standard deviation over 20 Monte Carlo trials.}
    \label{fig:lqr-comparison}
\end{figure}


\subsection{Online $\mathcal{H}_\infty$ Control}
\label{sec:hinf}

Another problem of broad interest and compatible with our setup is $\mathcal{H}_{\infty}$ robust control. Consider minimizing the $\ell_2$ gain from a disturbance input $w_t$ to the performance output $y_t = \sqrt{Q}x_t+\sqrt{R}u_t$. The associated model-based SDP formulation for the $\mathcal{H}_\infty$ problem is given by \cite{gahinet1994lmi}:
\begin{equation*}
    \begin{aligned}
        &\min_{X,Y,\gamma}~\gamma\\
        &\mathrm{subject~to}\\
        &
        \begin{bmatrix}
            X&(AX+BY)^\top&X^\top \sqrt{Q}&Y^\top \sqrt{R}&0\\
            AX+BY&X&0&0&I\\
            \sqrt{Q}X&0&\gamma I&0&0\\
            \sqrt{R}Y&0&0&\gamma I&0\\
            0&I&0&0&\gamma I
        \end{bmatrix}\succ 0,
    \end{aligned}
\end{equation*}
where the optimal controller is given by $K=YX^{-1}$.


Using similar parameterization approach as that for the LQR problem \eqref{eq:dee-lqr}, the direct data-driven formulation for $\mathcal{H}_\infty$ can be obtained as:
\begin{equation*}
    \begin{aligned}
        &\min_{P,\gamma}~\gamma\\
        &\mathrm{subject~to}\\
        &
        \begin{bmatrix}
        X_{0}P & (X_{1}P)^\top & (\sqrt Q\,X_{0}P)^\top & (\sqrt R\,U_{0}P)^\top & 0 \\
        X_{1}P & X_{0}P & 0 & 0 & I \\
        \sqrt Q\,X_{0}P & 0 & \gamma I & 0 & 0 \\
        \sqrt R\,U_{0}P & 0 & 0 & \gamma I & 0 \\
        0 & I & 0 & 0 & \gamma I
        \end{bmatrix}\succ 0.
    \end{aligned}
\end{equation*}
The control gain at iteration $t$ is given by $K_t=\overline{U}_{0,t}P_t(\overline{X}_{0,t}P_t)^{-1}$.

We test our algorithm on a challenging high-order system. In this case, we consider a controllable system with $n=14$ and $m=14$, where $B=I_{14}$ and $A$ is randomly generated. The length of the offline data is $t_0=2(14+14)=56$. The process noise is $w_t\sim \mathcal{N}(0,0.3I_{14})$, and the probing noise is $e_t\sim \mathcal{N}(0,0.5I_{14})$. To evaluate the performance of Algorithm \ref{al:online-control-first-order}, we conduct 20 Monte-Carlo experiments. In both figures, the $\mathcal{H}_{\infty}$ norm at time $t$ is computed as the one that corresponds to the closed-loop system $A+BK_t$. The regrets and optimality gaps are shown in Figures \ref{fig:hinf-reg} and \ref{fig:hinf-opt-gap}, respectively. It can be seen that our method improves the real $\mathcal{H}_{\infty}$ norm of the closed-loop system significantly through adaptation. Interestingly, the direct case show higher robustness than the indirect case as the variances of optimality gap and regret are slightly smaller. 

{To evaluate the computational efficiency of the proposed online updates, we compare the average runtime of one iteration of the proposed primal-dual method with repeatedly solving the corresponding SDP from scratch. Table~\ref{tab:runtime} reports the results under different dimensions. The proposed method consistently achieves a substantial reduction in computational time, with increasing speedups as the problem size grows. This demonstrates the scalability advantage of tracking the time-varying SDP through lightweight online updates instead of repeatedly solving full SDPs.

\begin{table}[t]
\centering
\caption{Average runtime per controller update for the $\mathcal{H}_\infty$ example.}
\label{tab:runtime}
\begin{tabular}{cccc}
\toprule
State dimension & Full SDP (s) & Ours (s)\\
\midrule
14 & 1.4349 & 0.0302 \\
20 & 7.4658 & 0.0860\\
40 & 184.5963 & 4.3422 \\
\bottomrule
\end{tabular}
\end{table}
}
\begin{figure}[!htbp]
    \centering
    \begin{subfigure}[t]{0.7\linewidth}
        \centering
        \includegraphics[width=\linewidth]{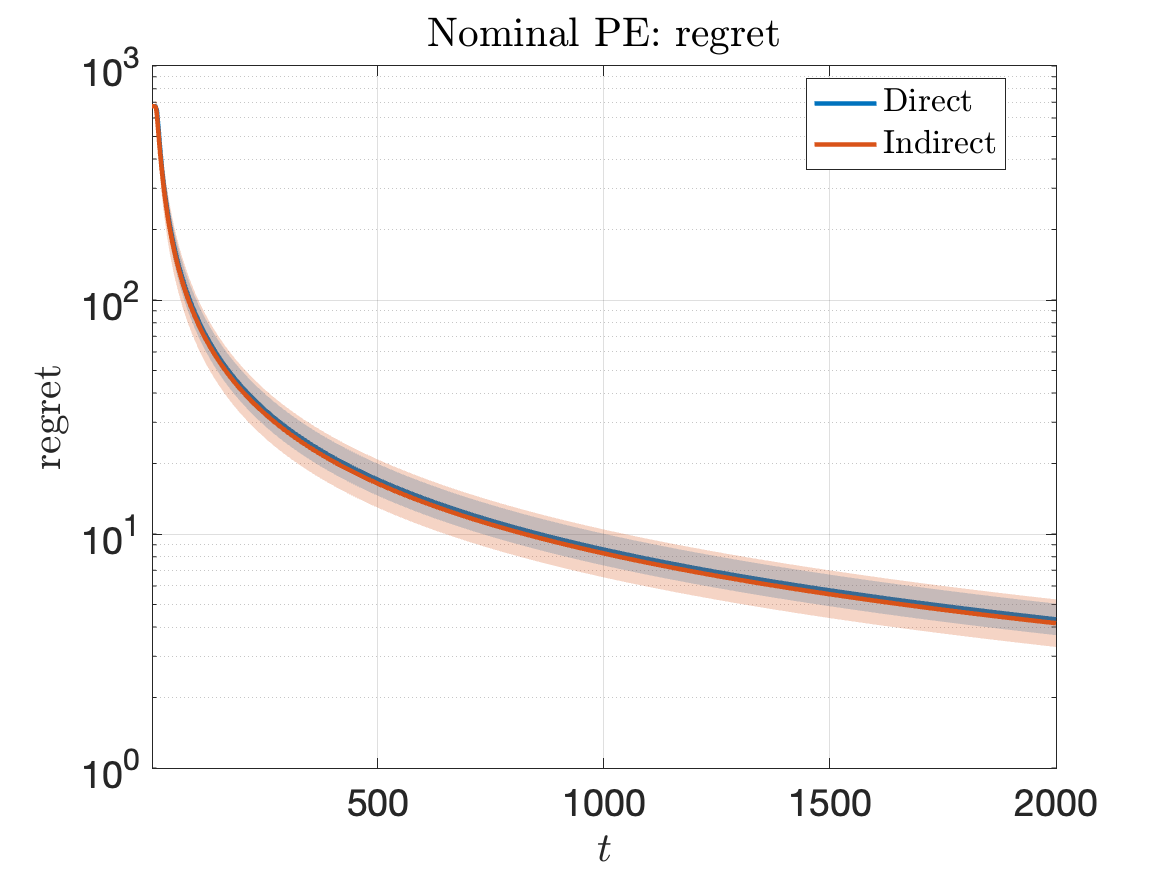}
        \caption{Regrets under the direct and indirect formulations.}
        \label{fig:hinf-reg}
    \end{subfigure}
    \hfill
    \begin{subfigure}[t]{0.7\linewidth}
        \centering
        \includegraphics[width=\linewidth]{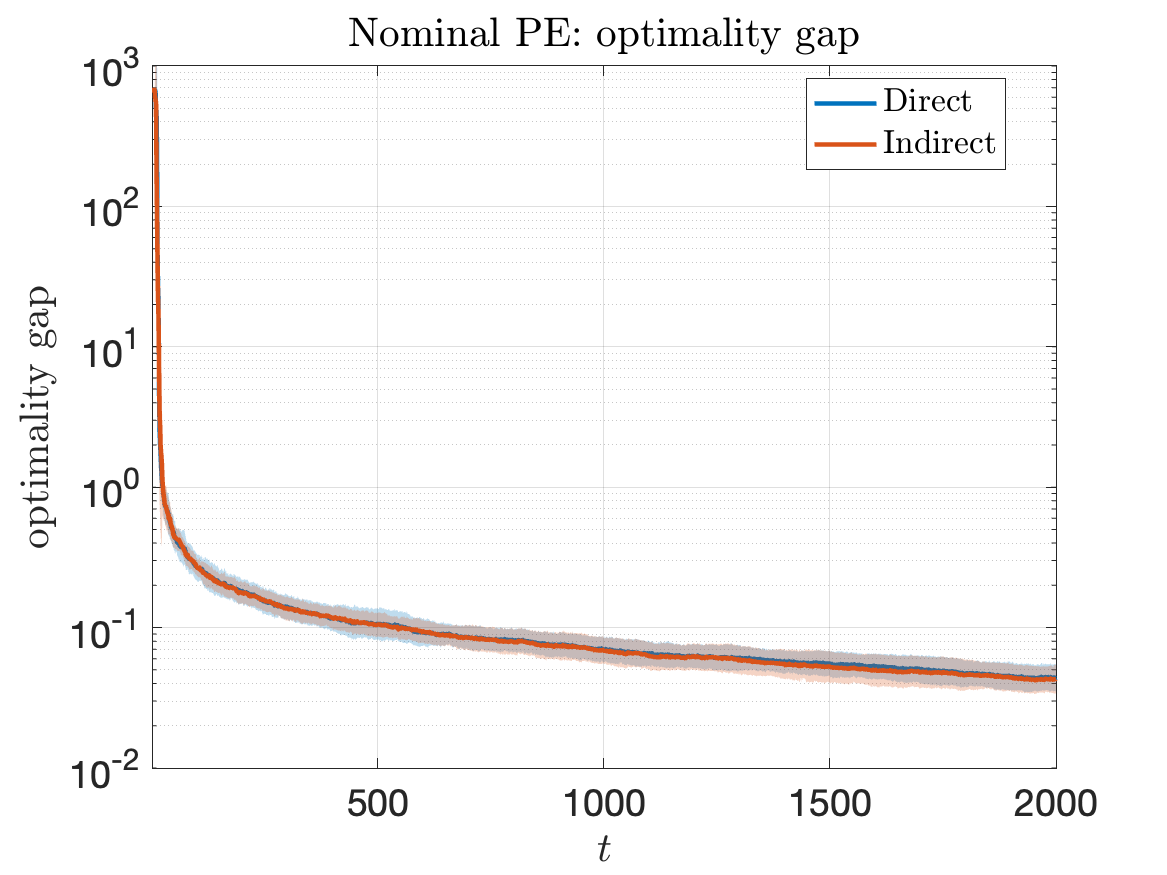}
        \caption{Optimality gaps under the direct and indirect formulations.}
        \label{fig:hinf-opt-gap}
    \end{subfigure}
    \caption{Convergence behavior for the $\mathcal{H}_\infty$ control problem, where the shaded areas represent one standard deviation over 20 Monte Carlo trials.}
    \label{fig:hinf-comparison}
\end{figure}


Another interesting observation is that for randomly generated system matrices $K_{t_0}$ is not a stabilizing gain in all cases. This is reflected in these figures as the $\mathcal{H}_{\infty}$ norm is $\infty$ initially, and set to $1000$. However, our method successfully stabilizes the system after finitely many iterations. This is demonstrated in Figure \ref{fig:spectral}. It can be seen that the spectral radius of $A+BK_{t_0}$ is greater than one, but decreases below one after around 20 iterations. The latter together with a bounded variation of $K_t$ is usually used to certify sequential strong stability \cite{cohen2018online}.

\begin{figure}[!htbp]
    \centering
    \includegraphics[width=0.7\linewidth]{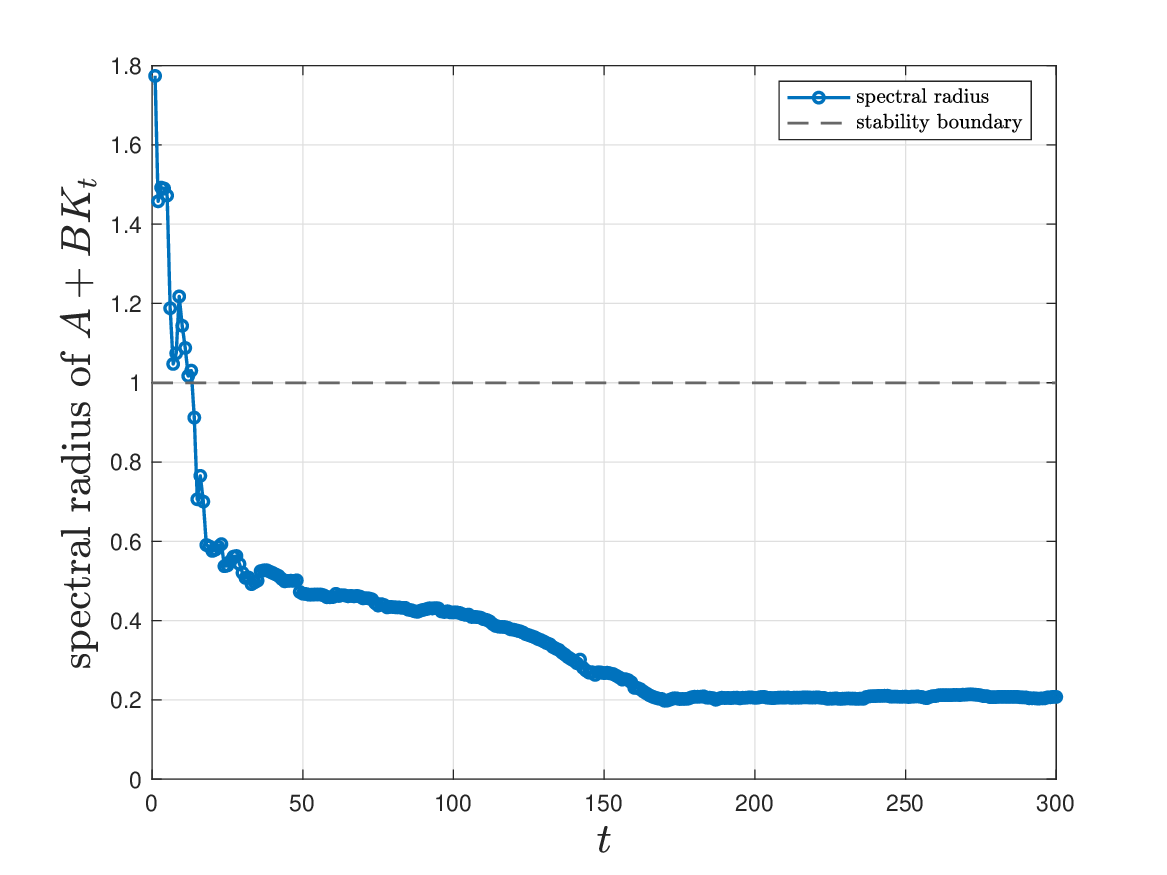}
    \caption{Spectral radius of the closed-loop system $A+BK_t$.}
    \label{fig:spectral}
\end{figure}

\subsection{Safe Exploration}
\label{sec:safety}
Another problem of interest is control under safety constraints $x_t\in\mathcal{X}$, $\forall t\ge 0$. We begin with the model-based formulation. The set $\mathcal{X}$ is assumed to be compact and to have nonempty interior:
\begin{equation}
    \begin{aligned}
        \mathcal{X}:=\bigcap_{i=1}^{N_\mathcal{X}}\{x\in\mathbb{R}^n:a_i^\top x+1\ge 0\},
    \end{aligned}
\end{equation}
where $a_i\in\mathbb{R}^n$, $i\in\{1,\ldots,N_\mathcal{X}\}$, are real vectors. One tractable sufficient condition for safety is to find a positively invariant set $\mathcal{P}\subseteq \mathcal{X}$ such that $\forall x_t\in \mathcal{P}$, we have $x_{t+1}\subseteq \mathcal{P}$. Consider an ellipsoidal parameterization
\begin{equation*}
    \mathcal{P}:=\{x\in\mathbb{R}^n:x^\top Px-\alpha\le 0\},
\end{equation*}
where $P\succ 0$ and $\alpha>0$. Using the model-based SDP formulation \cite{fochesato2025synthesis}, we propose the following direct data-driven SDP formulation for safe control:
\begin{equation}\label{eq:safety-program}
    \begin{aligned}
        \max_{Y}~&\mathrm{Tr}(\overline{X}_0Y)\\
        \mathrm{subject~to}~&
        \begin{bmatrix}
            \overline{X}_0Y&Y^\top \overline{X}_1\\
            \overline{X}_1Y&\overline{X}_0Y
        \end{bmatrix}\succeq 0,\\
        &
        \begin{bmatrix}
            1&x_0^\top\\
            x_0&\overline{X}_0Y
        \end{bmatrix}\succeq 0,\\
        &
        \begin{bmatrix}
            \overline{X}_0Y&\overline{X}_0Ya_i\\
            a_i^\top Y^\top \overline{X}_0^\top &1
        \end{bmatrix}\succeq 0,\quad i=1,\ldots,N_{\mathcal{X}},
    \end{aligned}
\end{equation}
where $K=\overline{U}_0Y(\overline{X}_0Y)^{-1}$.

In the experiment, the system dynamics are given by
\begin{equation}
A=
\begin{bmatrix}
    3&3\\
    2&2
\end{bmatrix},
\quad
B=
\begin{bmatrix}
    1&0\\
    0&3
\end{bmatrix}.
\end{equation}
The safe set is defined by a box with $a_1=[0\quad 0.2]^\top$, $a_2=[0\quad -0.2]^\top$, $a_3=[0.2\quad 0]^\top$, and $a_4=[-0.2\quad 0]^\top$. The process and probing noises are both Gaussian, with $w_t\sim \mathcal{N}(0,0.1I_2)$ and $e_t\sim \mathcal{N}(0,0.1I_2)$. The length of the offline data is $t_0=6$. In this setting, since the noise is unbounded, only probabilistic safety guarantees can be obtained. For a more detailed analysis, interested readers are referred to \cite{fochesato2025synthesis}.

For comparison, we compare the closed-loop trajectory of the system controlled by Algorithm \ref{al:online-control-first-order} with that controlled by a standard LQR controller with $Q=R=I_2$. Simulation results are provided in Figure \ref{fig:safety}. It can be seen that the LQR controller steers the system outside the box, while our method ensures safety for all time.

\begin{figure}[!htbp]
    \centering
    \includegraphics[width=0.8\linewidth]{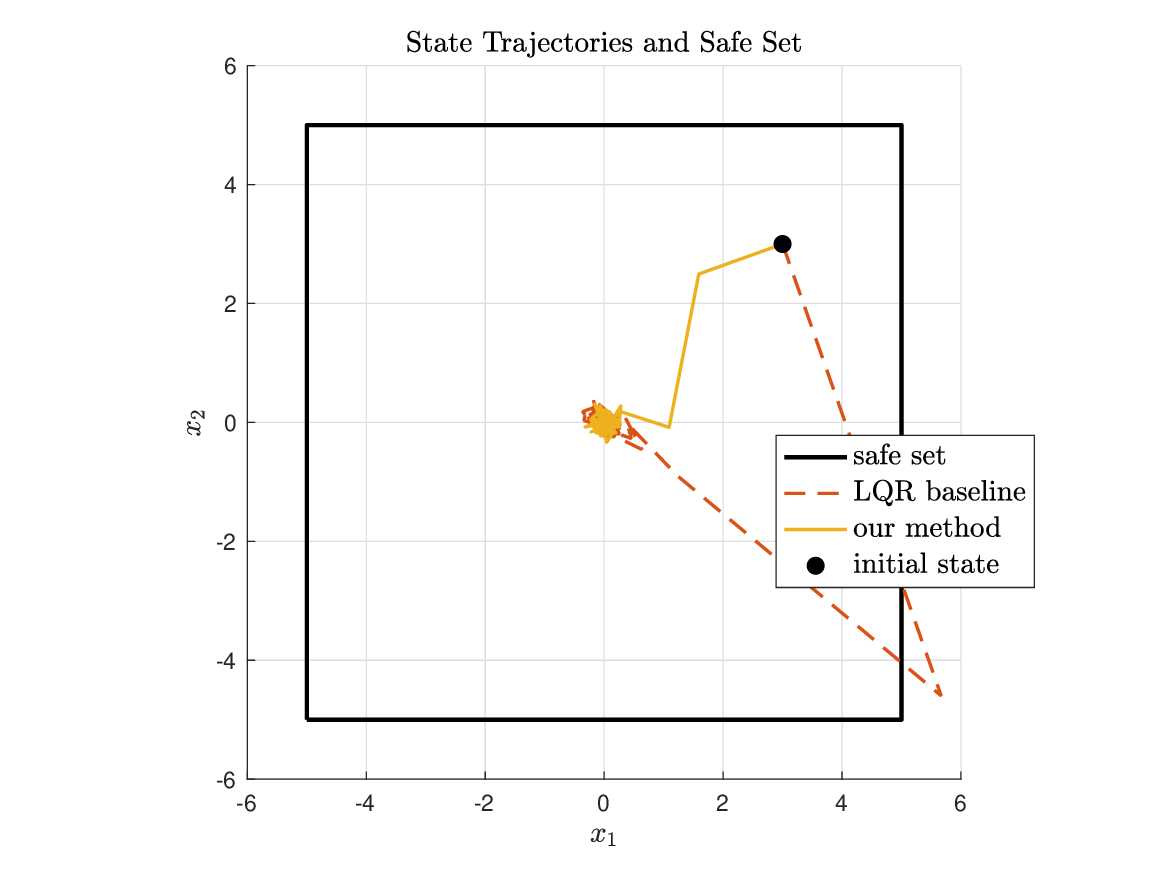}
    \caption{Comparison between our method and the LQR baseline.}
    \label{fig:safety}
\end{figure}





\section{Conclusion}

This paper proposed a data-enabled primal-dual framework for online learning optimal control policies via SDP formulations. By treating the data-driven synthesis problem as a time-varying SDP, the proposed method recursively updates the policy using lightweight primal-dual iterations rather than repeatedly solving full SDPs. The framework unifies direct and indirect data-driven formulations and applies to LQR, $\mathcal{H}_\infty$ control, and safety-critical control with state constraints. We introduced the Difference-of-Signal and Sim-to-Real Gap to characterize the data-driven SDP, and established local linear tracking and global ergodic convergence under persistency of excitation and standard SDP regularity conditions. Numerical results demonstrated improved control performance, support for general SDP constraints beyond policy-gradient methods, and substantial computational savings for recursive data-driven policy learning.

\section{Acknowledgement}
The authors acknowledge Dr. Yuwen Chen from Nvidia and Prof. Jaap Eising from the University of Groningen for their valuable inputs.

\appendix

\subsection{Proof of Lemma \ref{lem:dos-evolution}}

\begin{proof}
We first analyze $\mathrm{DOS}_{t}^d$ on the direct setting. Using Assumption \ref{ass:linear}, we have $\mathrm{DOS}_{t}^d=\|\mathcal{L}(D_t)-\mathcal{L}(D_{t-1})\|_F\le L_{\mathrm{DOS}}^{\mathrm d}
    \left\|\left[\overline X_{1,t}-\overline X_{1,t-1} \quad  \overline X_{0,t}-\overline X_{0,t-1} \quad \overline U_{0,t}-\overline U_{0,t-1}\right]
    \right\|_F$ for some $L_{\mathrm{DOS}}^d>0$. It holds that $\overline{X}_{1,t}-\overline{X}_{1,t-1}=\frac{1}{t}(x_{t}D_{t-1}^\top -\overline{X}_{1,t-1})=\frac{1}{t}(x_{t}D_{t-1}^\top - \frac{1}{t-1}\sum_{k=0}^{t-2}x_{k+1}D_{k}^\top)$. To quantify $\|\overline{X}_{1,t}-\overline{X}_{1,t-1}\|_F$, it suffices to quantify $\|x_{k+1}D_{k}^\top \|_F$. By Cauchy--Schwarz, we have $\mathbb{E}\|x_{k+1}D_{k}^\top \|_F\le c\left(\mathbb{E}\|x_{k+1}\|^2\right)^{1/2}\left(\mathbb{E}\|D_k\|^2\right)^{1/2}\le C$ for some $c$ and $C$. Here, we have used the facts that i) $\mathbb{E}\|x_k\|^2<\infty$ from closed-loop stability \eqref{eq:second-moment}; and ii) $\mathbb{E}\|u_k\|^2=\mathbb{E}\|K_kx_k+e_k\|^2 < \infty$. Therefore, by Markov's inequality, there is $C>0$ so that for any $\beta >0$:
\begin{equation*}\mathbb{P}\left(\|x_{k+1}D_{k}^\top \|_F > \beta \right)\le \frac{C}{\beta}
\end{equation*}
This implies that $\|\overline{X}_{1,t}-\overline{X}_{1,t-1}\|_F=\mathcal{O}_p(1/t)$. Similarly, we can prove that $\|\overline{X}_{0,t}-\overline{X}_{0,t-1}\|_F=\mathcal{O}_p(1/t)$ and $\|\overline{U}_{0,t}-\overline{U}_{0,t-1}\|_F=\mathcal{O}_p(1/t)$, which then proves \eqref{eq:direct-dos}. 

Under Assumption \ref{ass:pe}, the ordinary least-squares problem \eqref{eq:ols} admits a unique solution $(\hat B_t, \hat A_t)=\overline{X}_{1,t}\Phi_t^{-1}$. Then, Assumption \ref{ass:linear} implies that $$\mathrm{DOS}_t^i=\|\mathcal{L}(\hat A_t,\hat B_t)-\mathcal{L}(\hat A_{t-1},\hat B_{t-1})\|_F \le L_{\mathrm{DOS}}^i\left\|
        \overline X_{1,t}\Phi_{t}^{-1}
        -
        \overline X_{1,t-1}\Phi_{t-1}^{-1}
    \right\|_F$$ for some $L_{\mathrm{DOS}}^i>0$. The rest of the proof for \eqref{eq:indirect-dos} follows a similar procedure to that for the direct formulation, and is therefore omitted. 
\end{proof}

\subsection{Proof of Lemma \ref{lem:srg-evolution}}

\begin{proof}
    We first analyze $\mathrm{SRG}_t^d$ in the direct setting. Similar to the analysis for $\mathrm{DOS}_t^d$ in the proof of Lemma \ref{lem:dos-evolution}, it can be shown that $\mathrm{SRG}_t^d \le L_{\mathrm{SRG}}^d \|\overline{W}_{0,t}\|_F$ for some $L_{\mathrm{SRG}}^d>0$. From the expression of $\|\overline{W}_{0,t}\|_F$, we have 
    \begin{equation}
        \begin{aligned}
\mathbb E\|\overline{W}_{0,t}\|_F^2=&\mathbb E\|\frac{1}{t}\sum _{k=0}^{t-1}w_kD_k^\top \|_F^2\\
=&
\frac{1}{t^2}
\sum_{k=0}^{t-1}
\mathbb E\|w_kD_k^\top\|_F^2
+
\frac{1}{t^2}
\sum_{\substack{i,j=0\\i\neq j}}^{t-1}
\mathbb E\operatorname{tr}
\left(
w_iD_i^\top D_jw_j^\top
\right).
        \end{aligned}
    \end{equation}
Since $w_k\sim \mathcal{N}(0,I_n)$, we have $\mathbb{E}\|w_k\|^2=n$. Also, $D_k$ is generated by the data up to time $k$, while $w_k$ is the new process noise at time $k$. Hence, $w_k$ is independent of $D_k$. For $i\ne j$, the cross term vanishes due to $w_k$ being i.i.d. Therefore, $\mathbb E\|\overline W_{0,t}\|_F^2
=
\frac{1}{t^2}
\sum_{k=0}^{t-1}
\mathbb E\|w_kD_k^\top\|_F^2.$ Since $\|w_kD_k^\top\|_F^2
\le
\|w_k\|^2\|D_k\|^2$, and $w_k$ is independent of $D_k$, we have 
\begin{equation}
    \mathbb E\|w_kD_k^\top\|_F^2
\le
\mathbb E\|w_k\|^2 \, \mathbb E\|D_k\|^2
\le
nC_D,
\end{equation}
where $C_D:=\sup_{k\ge0}\mathbb E\|D_k\|^2<\infty .$ Hence, we have $\mathbb E\|\overline W_{0,t}\|_F
\le
\sqrt{\frac{1}{t^2}
\sum_{k=0}^{t-1} nC_D}
=
\sqrt{\frac{nC_D}{t}}.$ By Markov's inequality, for any $\beta>0$ we have 
\begin{equation}
    \mathbb P
\left(
\sqrt t\,\|\overline W_{0,t}\|_F>\beta
\right)
\le
\frac{t\,\mathbb E\|\overline W_{0,t}\|_F^2}{\beta^2}
\le
\frac{nC_D}{\beta^2}.
\end{equation}
Thus $\sqrt t\,\|\overline W_{0,t}\|_F
=
\mathcal{O}_p(1)$, which is equivalent to $\|\overline W_{0,t}\|_F
=
\mathcal{O}_p\!\left(\frac{1}{\sqrt t}\right).$ This proves \eqref{eq:direct-srg}. 

The proof for the indirect case \eqref{eq:indirect-srg} follows a similar procedure, by noting that $\|\Phi_t^{-1}\|_F=\mathcal{O}_p(1)$. 
\end{proof}

\subsection{Proof of Theorem \ref{th:opt-gap-linear}}
\begin{proof}
Define a square matrix
    \begin{equation}\label{eq:def-Z}
H_t:=\operatorname{vec}^{-1}\!\left(s_t-\lambda_t/\rho\right),
    \end{equation}
Using this definition, let $H^*_t$ be the matrix that corresponds to $s^*(A_{1,t})$ and $\lambda^*(A_{1,t})$. Given that every $t=t_0,\ldots,T$ Assumption \ref{ass:regularity} holds, then \cite[Theorem 3]{kang2025local} shows that for $H^*_t$, there exist $\varepsilon^{H}_t>0$ and $\sigma_t\in(0,1)$ such that
\begin{equation}\label{eq:q-contraction}
    \|H_{t+1}-H^*_t\|_F\le \sigma_t\|H_t-H_t^*\|_F,\quad \forall H_t\in\mathcal{B}(H_t^*,\varepsilon^{H}_t)
\end{equation}

The key difference between the present setting and the static setting in \cite{kang2025local} is that $A_{1,t}$, and hence the fixed point $H_t^*$, varies with time. Therefore, \eqref{eq:q-contraction} cannot be applied directly as a convergence result to a fixed point. Instead, we use it as a local tracking estimate and quantify the effect of the movement of $H_t^*$.

Define the minimal radius and maximal contraction rate by 
\begin{equation}\label{eq:quantity-def}
    \varepsilon_{T}^H:=\min_{t\in\{t_0,\dots,t_0+T\}}\varepsilon_t^H,\quad \sigma_{T}:=\max_{t\in\{t_0,\dots,t_0+T\}}\sigma_t
\end{equation}
Then, it holds that $\varepsilon_{T}^H>0$ and $\sigma_{T}\in(0,1)$. 
From \eqref{eq:q-contraction} and \eqref{eq:quantity-def}, we deduce that 
\begin{equation}\label{eq:set-inclusion}
\begin{aligned}
     &H_{t+1}\in \mathcal{B}(H_t^*,\sigma_t\varepsilon_T^H)\subseteq \mathcal{B}(H_t^*,\sigma_{T}\varepsilon_{T}^H),\\
     &\forall H_t\in\mathcal{B}(H_t^*,\varepsilon_T^H),t\in\{t_0,\dots,t_0+T-1\}
\end{aligned}
\end{equation}
This is because \eqref{eq:q-contraction} holds on $\mathcal{B}(H_t^*,\varepsilon_T^{H})\subseteq \mathcal{B}(H_t^*,\varepsilon_t^{H})$.

Using this set inclusion relation \eqref{eq:set-inclusion}, we seek to find the lower bound $\varepsilon_{T}^{\mathrm{DOS}}$ for $\mathrm{DOS}_t$ such that $H_{t}\in\mathcal{B}(H_t^*,\varepsilon_{T}^H),\forall t\in\{t_0,\dots,t_0+T\}$. Note that $\mathrm{DOS}_t$ affects both $H_t^*$ and $H_t$ in the relations. 

From \eqref{eq:def-Z}, we can deduce the following relation
\begin{equation*}
\begin{aligned}
     \|H_t^*-H_{t+1}^*\|_F\le &\|s_t^*-s_{t+1}^*\|_F+\|\lambda_t^*-\lambda_{t+1}^*\|_F/\rho\\
     \le &\|\omega_{t}^*-\omega_{t+1}^*\|_F+\|\omega_t^*-\omega_{t+1}^*\|_F/\rho
\end{aligned}
\end{equation*}
Using this relation and Lemma \ref{lem:smoothness}, we obtain that
\begin{equation*}
\begin{aligned}
    &\|H_t^*-H_{t+1}^*\|_F \le (1+1/\rho) L\|A_{1,t+1}-A_{1,t}\|_F\\
    &=(1+1/\rho)L\mathrm{DOS}_{t+1}\le (1+1/\rho)L\varepsilon_T^{\mathrm{DOS}}
\end{aligned}
\end{equation*}

Define
\begin{equation}\label{eq:snr-bound}
\varepsilon_T^{\mathrm{DOS}} = \frac{(1-\sigma_T)\varepsilon_T^H}{(1+1/\rho)L}
\end{equation}
and substitute the value of $\varepsilon_{T}^{\mathrm{DOS}}$ into the above inequality. We obtain
\begin{equation}\label{eq:center-distance}
    \|H_{t}^*-H_{t+1}^*\|\le (1-\sigma_{T})\varepsilon_{T}^H,\quad \forall t\in\{t_0,\dots,t_0+T-1\}
\end{equation}
From this inequality, we deduce that 
\begin{equation}\label{eq:second-last-inequality}
    \mathcal{B}(H_t^*,\sigma_{T} \varepsilon_{T}^H)\subseteq \mathcal{B}(H_{t+1}^*,\varepsilon_{T}^H)\subseteq \mathcal{B}(H_{t+1}^*,\varepsilon_{t+1}^H)
\end{equation}
The first set inclusion relation follows from the fact that the distance between the centers of the two balls $\mathcal{B}(H_t^*,\sigma_{T} \varepsilon_{T}^H)$ and $\mathcal{B}(H_{t+1}^*,\varepsilon_{T}^H)$ is less than or equal to $(1-\sigma_{T})\varepsilon_{T}^H$ from \eqref{eq:center-distance}, and $(1-\sigma_{T})\varepsilon_{T}^H+\sigma_{T}\varepsilon_{T}^H=\varepsilon_{T}^H$. The second set inclusion relation follows from the fact that $\varepsilon_{t+1}^H \ge \varepsilon_{T}^H$. 

Combining \eqref{eq:second-last-inequality} and \eqref{eq:set-inclusion}, we obtain
\begin{equation}
\begin{split}
    &H_{t+1}\in \mathcal{B}(H_{t+1}^*,\varepsilon_{t+1}^H),\quad \\
    &\forall H_t\in\mathcal{B}(H_t^*,\varepsilon_t^H),t\in\{t_0,\dots,t_0+T-1\}
\end{split}
\end{equation}
Therefore, if $H_{t_0}\in\mathcal{B}(H_{t_0}^*,\varepsilon_{t_0}^H)$, we have $H_{t+1}\in \mathcal{B}(H_{t+1}^*,\varepsilon_{t+1}^H)$ for any $t\in\{t_0,\dots,t_0+T-1\}$. This means that $H_t$ remains in the local region where the local contraction property \eqref{eq:q-contraction} holds, which can then be used to bound $H_t$. From the relation between $\omega$ and $H$ in \eqref{eq:def-Z}, it suffices to have $H_{t_0}\in\mathcal{B}(H_{t_0}^*,\varepsilon_{t_0}^H)$ from $\omega_{t_0}\in\mathcal{B}(\omega_{t_0}^*,\varepsilon_T^\omega)$ for some $\varepsilon_T^\omega>0$. 

Using the contraction property \eqref{eq:q-contraction}, we are now ready to analyze the optimality gap. For $t\in\{t_0,\dots,t_0+T-1\}$, using \eqref{eq:q-contraction} and replacing $\sigma_t$ with $\sigma_{T}$ implies:
\begin{equation}\label{eq:th2eq2}
    \begin{aligned}
    &\|H_{t+1}-H_{t+1}^*\|_F\le \|H_{t+1}-H_t^*\|_F+\|H_t^*-H_{t+1}^*\|_F\\
    &\le \sigma_T\|H_{t}-H_t^*\|_F+\|H_t^*-H_{t+1}^*\|_F\\
    &\le \sigma_T\|H_t-H_t^*\|_F+(1+1/\rho)\|\omega_t^*-\omega_{t+1}^*\|_F\\
    &\le \sigma_T\|H_{t}-H_t^*\|_F+L(1+1/\rho)\mathrm{DOS}_{t+1}
    \end{aligned}
\end{equation}
Here, the first inequality is obtained via the standard triangle inequality, and the second inequality follows from the contraction \eqref{eq:q-contraction} after replacing $\sigma_t$ by the upper bound $\sigma_T$. The third inequality comes from the definition of $H_t$ in \eqref{eq:def-Z}, and the fourth one is obtained by using Lemma \ref{lem:smoothness} and recalling the definition of $\mathrm{DOS}_t$ in \eqref{eq:dos}.

Using the comparison lemma and considering the corresponding proportional sequence, we immediately obtain 
\begin{equation}\label{eq:th2eq3}
\begin{aligned}
        &\|H_{t+1}-H^*_{t+1}\|_F\le \sigma_{T}^{t+1-t_0} \nu_{1}+\nu_{2}\sum_{k=t_0}^{t+1}\sigma_{T}^{t-k}\mathrm{DOS}_k,\\
        &\forall t\in\{t_0,\dots,t_0+T-1\}
\end{aligned}
\end{equation}
for some $0<\nu_{1},\nu_{2}<\infty$. Under the assumption that $\mathrm{DOS}_t\le \varepsilon_{T}^{\mathrm{DOS}}$, we immediately obtain that $\|H_{t+1}\|_F$ is bounded for any $t\in[t_0,t_0+T-1]$. 

To establish local linear convergence of the optimality gap, we first use the inequality \eqref{eq:th2eq3} to prove boundedness of $(z_t,s_t,\lambda_t)$. From the definition of $H_t$, the Moreau decomposition \cite[Chapter 2.5]{parikh2014proximal} shows that
\begin{equation}
    s_t=\mathrm{vec}(\Pi_{\mathbb{S^+}}(H_t)),\quad \lambda_t=\rho\mathrm{vec}(\Pi_{\mathbb{S}^+}(-H_t)) 
\end{equation}
From the non-expansiveness of the cone projection operator, we immediately obtain that both $\|\lambda_t\|_F$ and $\|s_t\|_F$ are bounded, because they satisfy 
\begin{equation}\label{eq:56}
    \|\lambda_t-\lambda_{t}^*\|_F\le  \|H_{t}-H_t^*\|_F/\rho,\quad \|s_t-s_t^*\|_F\le \|H_{t}-H_t^*\|_F
\end{equation}
Using \eqref{eq:th2eq3}, we obtain that
\begin{equation}\label{eq:slambda}
    \begin{aligned}
        &\|\lambda_t-\lambda_t^*\|_F \le \sigma_T^{t-t_0}\nu_1/\rho + \nu_2\sum_{k=t_0}^t \sigma_T^{t-k}\mathrm{DOS}_k/\rho\\
        &\|s_t-s_t^*\|_F \le \sigma_T^{t-t_0}\nu_1+\nu_2\sum_{k=t_0}^t \sigma_T^{t-k}\mathrm{DOS}_k
    \end{aligned}
\end{equation}

We then characterize $z_t$. For problem \eqref{eq:dee-sdp2}, the $z$-update in \textbf{Step 3} reads as
\begin{equation*}
    z_{t+1}=\mathop{\arg\min}_z~c^\top z+\frac{\rho}{2}\|A_{1,t}z+A_2s_t-b-\lambda_t/\rho\|_F^2\end{equation*}
The optimality condition implies 
\begin{equation}\label{eq:th2eq4}
    z_{t+1}
=
(\rho A_{1,t}^\top A_{1,t})^{-1}(\rho A_{1,t}^\top(b-A_2s_t)+A_{1,t}^\top\lambda_t-c).
\end{equation}
From Assumption \ref{ass:pe}, we have that $A_{1,t}$ always has full column rank and $\underline{\sigma}(A_{1,t})\ge \underline{\gamma}$, which implies $\|(A_{1,t}^\top A_{1,t})^{-1}\|_2\le \underline{\gamma}^{-2}$. This implies boundedness of $z_{t+1}$. 

We then calculate the upper bound for $\|z_{t+1}-z^*_{t+1}\|_F$. Using \eqref{eq:th2eq4}, we have
\begin{equation*}
    \rho A_{1,t}^\top A_{1,t}z_{t+1}
=
\rho A_{1,t}^\top(b-A_2s_t)+A_{1,t}^\top\lambda_t-c.
\end{equation*}
Subtracting $\rho A_{1,t}^\top A_{1,t}z^*_{t+1}$ from both sides, we have:
\begin{equation}\label{eq:th2eq5}
\begin{aligned}
    &\rho A_{1,t}^\top A_{1,t}(z_{t+1}-z^*_{t+1})=\rho A_{1,t}^\top(b-A_2s_t-A_{1,t}z^*_{t})\\
    &+A_{1,t}^\top\lambda_t-c- \rho A_{1,t}^\top A_{1,t}(z^*_{t+1}-z^*_t)
\end{aligned}
\end{equation}
The term $b-A_2s_t-A_{1,t}z^*_{t}$ can be rewritten as
\begin{equation*}
    b-A_2s_t-A_{1,t}z^*_{t}=A_2(s_t^*-s_t)
\end{equation*}
where we have used primal feasibility $A_{1,t}z_t^*+A_2s_t^*=b$. Substituting the above equation into \eqref{eq:th2eq5}, we obtain
\begin{equation}
\begin{aligned}
    &\rho A_{1,t}^\top A_{1,t}(z_{t+1}-z^*_{t+1})=\\
    &\rho A_{1,t}^\top A_2(s_t^*-s_t)-\rho A_{1,t}^\top A_{1,t}(z_{t+1}^*-z_t^*)+A_{1,t}^\top\lambda_t-c
\end{aligned}
\end{equation}
Finally, using the target stationarity condition $c=A_{1,t}^\top \lambda^*_{t}$ yields
\begin{equation*}
    A_{1,t}^\top \lambda_t-c=A_{1,t}^\top (\lambda_t-\lambda_t^*)
\end{equation*}
Substituting this into \eqref{eq:th2eq5}, we obtain 
\begin{equation}\label{eq:60}
\begin{aligned}
    &\rho A_{1,t}^\top A_{1,t}(z_{t+1}-z^*_{t+1})=\\
    &\rho A_{1,t}^\top A_2(s_t^*-s_t)-\rho A_{1,t}^\top A_{1,t}(z_{t+1}^*-z_t^*)-A_{1,t}^\top (\lambda_t^*-\lambda_t)
\end{aligned}
\end{equation}

Using \eqref{eq:slambda}, and considering $\|z_{t+1}^*-z_t^*\|_F\le \|\omega_{t+1}^*-\omega_t^*\|_F\le L\|A_{1,t+1}-A_{1,t}\|_F=L\mathrm{DOS}_{t+1}$, and recalling that $\|(A_{1,t}^\top A_{1,t})^{-1}\|_2\le \underline{\gamma}^{-2}$, we obtain the following bound
\begin{equation}\label{eq:z-convergence}
\begin{aligned}
     &\|z_{t+1}-z^*_{t+1}\|_F\le \nu_{3}\sigma_{T}^{t+1-t_0}+\nu_{4}\sum_{k=t_0}^{t+1}\sigma_{T}^{t-k}\mathrm{DOS}_k+\nu_{5}\mathrm{DOS}_{t+1}\\
     & \le \nu_{3}\sigma_{T}^{t+1-t_0}+\nu_{6}\sum_{k=t_0}^{t+1}\sigma_{T}^{t-k}\mathrm{DOS}_k,\quad \forall t\in\{t_0,\dots,t_0+T-1\}
\end{aligned}
\end{equation}
for some $0<\nu_3,\nu_4,\nu_5,\nu_6<\infty$. We are now ready to quantify the optimality gap between the runtime cost $f_1(z_{t+1})$ and the optimal cost $f_1(\tilde z^*_{t+1})$, which corresponds to the noiseless data $A_{1,t+1}^{\mathrm{nl}}$. We can then quantify $c^\top (z_{t+1}^*-\tilde z^*_{t+1})$ by SRG:
\begin{equation}\label{eq:46}
\begin{aligned}
      &|c^\top (z_{t+1}^*-\tilde z^*_{t+1})|=|c^\top (z_{t+1}^*-\tilde{z}_{t+1}^*)|\\
      &\le\|c\|_FL\|A_{1,t+1}-{A}_{1,t+1}^{\mathrm{nl}}\|_F\le C_3\mathrm{SRG}_{t+1}
\end{aligned}
\end{equation}
for some $0<C_3<\infty$. The first inequality is obtained by using Lemma \ref{lem:smoothness}. The second inequality is obtained via Assumption \ref{ass:linear}, which assumes that $A_{1,t+1}$ is linear in $\bar {X}_{1,t+1}$. The difference between $A_{1,t+1}$ and ${A}_{1,t+1}^{\mathrm{nl}}$ is therefore linear in $\mathrm{SRG}_{t+1}$.

Therefore, from \eqref{eq:z-convergence} and \eqref{eq:46}, we have
\begin{equation}\label{eq:64}
\begin{aligned}
    &|f_1(z_{t+1})-f_1(z^*)| \le \|c\|_F\nu_{3}\sigma_{T}^{t+1-t_0}+\\
    &\|c\|_F\nu_{6}\sum_{k=t_0}^{t+1}\sigma_{T}^{t-k}\mathrm{DOS}_k+C_3\mathrm{SRG}_{t+1}\\
    &\forall t\in\{t_0,\dots,t_0+T-1\}
\end{aligned}
\end{equation}
By choosing $C_1=\|c\|_F\nu_{3}$ and $C_2=\|c\|_F\nu_{6}$, we obtain \eqref{eq:opt-gap-linear}. Here, we have used the fact that $f_2(s_{t+1})=f_2^*=0$, because of \textbf{Step 4} of Algorithm \ref{al:online-control-first-order}.
\end{proof}



\subsection{Proof of Proposition \ref{prop:constraint}}
\begin{proof}
    Let 
    \begin{equation}
        r_{t+1}={A}^{\mathrm{nl}}_{1,t}z_{t+1}+A_2s_{t+1}-b
    \end{equation}
We decompose $r_{t+1}=( A_{1,t}^{\mathrm{nl}}-A_{1,t})z_{t+1}+A_{1,t}z_{t+1}+A_2s_{t+1}-b$ and apply the triangular inequality:
\begin{equation}\label{eq:asd}
    \|r_{t+1}\|_F\le \|A_{1,t}-A_{1,t}^{\mathrm{nl}}\|_F\|z_{t+1}\|_F+\|A_{1,t}z_{t+1}+A_2s_{t+1}-b\|_F
\end{equation}
From \textbf{step 5} of Algorithm \ref{al:online-control-first-order}, we have $A_{1,t}z_{t+1}+A_2s_{t+1}-b=\frac{1}{\rho}(\lambda_{t+1}-\lambda_t)$. Using this into \eqref{eq:asd}, we obtain 
\begin{equation}\label{eq:residual}
    \|r_{t+1}\|_F
\le
\|A_{1,t}^{\mathrm{nl}}-A_{1,t}\|_F\,\|z_{t+1}\|_F
+
\frac{1}{\rho}\|\lambda_{t+1}-\lambda_t\|_F
\end{equation}
We then bound the two terms separately. For the first term, using \eqref{eq:z-convergence}, we have 
\begin{equation}\label{eq:bound1}
\begin{aligned}
        &\|A_{1,t}^{\mathrm{nl}}-A_{1,t}\|_F\|z_{t+1}\|_F
\le\\
&\mathrm{SRG}_t\|z^*_{t+1}\|_F
+
\nu_{3}\sigma_{T}^{t+1-t_0}+\nu_{6}\sum_{k=t_0}^{t+1}\sigma_{T}^{t-k}\mathrm{DOS}_k
\end{aligned}
\end{equation}

For the second term, it holds that
\begin{equation}\label{eq:bound2}
\begin{split}
        &\|\lambda_{t+1}-\lambda_t\|_F
\le
\|\lambda_{t+1}-\lambda_t^\star\|_F+\|\lambda_t-\lambda_t^\star\|_F
\le\\
&
\|H_{t+1}-H_t^\star\|_F/\rho + \|H_t-H_t^\star\|_F/\rho\\
&\le \frac{1+\sigma_T}{\rho}\left(\sigma_{T}^{t-t_0} \nu_{1}+\nu_{2}\sum_{k=t_0}^t\sigma_{T}^{t-k}\mathrm{DOS}_k\right) 
\end{split}
\end{equation}
The first inequality is standard triangular inequality, the second one comes from \eqref{eq:56} and the third one uses \eqref{eq:q-contraction} and \eqref{eq:th2eq3}. 

Using the inequalities in \eqref{eq:bound1} and \eqref{eq:bound2} into \eqref{eq:residual}, we obtain \eqref{eq:50} for some $0<C_5,C_6,C_7<\infty$. 
\end{proof}

\subsection{Proof of Theorem \ref{th:opt-gap}}

The following lemma is the key preliminary result for convergence, showing that that the sequence $s_t$ and $\lambda_t$ can be upper bounded by $\mathrm{DOS}_t$.  

\begin{lemma}\label{lem:boundedness}
    Let Assumptions \ref{ass:pe} and \ref{ass:regularity} hold. Then, for any $T\ge t_0$, there exist constants $0<\beta_1,L_\mathrm{norm}<\infty$ such that
    \begin{equation}\label{eq:lem4}
        \|s_{T+1}\|_F,\|\lambda_{T+1}\|_F\le \beta_1+L_{\mathrm{norm}}\sum_{t=t_0}^T \mathrm{DOS}_t
    \end{equation}
\end{lemma}

\begin{proof}
The proof builds upon the non-expansiveness of the fixed-iteration operator.
Define
\begin{equation*}
    G:=\begin{bmatrix}
        0&0&0\\0&\rho A_2^\top A_2&0\\0&0&\frac{1}{\rho}I
    \end{bmatrix}, \omega^*(A_{1,t}):=\begin{bmatrix}
        z^*(A_{1,t})\\
        s^*(A_{1,t})\\
        \lambda^*(A_{1,t})

    \end{bmatrix},\omega_t:=\begin{bmatrix}
        z_t\\s_t\\\lambda_t
    \end{bmatrix}
\end{equation*}
From \cite[Lemma 3.2]{he20121}, it holds that
\begin{equation*}
    \|\omega_{t+1}-\omega^*(A_{1,t})\|_G^2\le\|\omega_t-\omega^*(A_{1,t})\|_G^2-\|\omega_{t}-\omega_{t+1}\|_G^2
\end{equation*}
Dropping the term $\|\omega_t-\omega_{t+1}\|_G^2$, we have 
\begin{equation}\label{eq:contractivity}
    \|\omega_{t+1}-\omega^*(A_{1,t})\|_G\le\|\omega_t-\omega^*(A_{1,t})\|_G
\end{equation}
From the triangle inequality for Euclidean norms, we have that
\begin{equation*}
\begin{aligned}
    &\|\omega_{t+1}-\omega^*(A_{1,t+1})\|_G\\
    \le&\|\omega_{t+1}-\omega^*(A_{1,t})\|_G+\|\omega^*(A_{1,t+1})-\omega^*(A_{1,t})\|_G
\end{aligned}
\end{equation*}
Combining this relation with inequality \eqref{eq:contractivity}, we have 
\begin{equation*}
\begin{aligned}
    &\|\omega_{t+1}-\omega^*(A_{1,t+1})\|_G\le\|\omega_t-\omega^*(A_{1,t})\|_G\\
    &+\|\omega^*(A_{1,t+1})-\omega^*(A_{1,t})\|_G
\end{aligned}
\end{equation*}
Summing it from $t=t_0$ to $t=T$ yields:
\begin{equation}\label{eq:lem4eq1}
\begin{aligned}
    &\|\omega_{T+1}-\omega^*(A_{1,T+1})\|_G\\
    \le &\|\omega_{t_0}-\omega^*(A_{1,t_0})\|_G+\sum_{t=t_0}^T\|\omega^*(A_{1,t+1})-\omega^*(A_{1,t})\|_G\\
\end{aligned}
\end{equation}
From Lemma \ref{lem:smoothness} and the equivalence of Euclidean norms, we have 
\begin{equation*}
\begin{aligned}
    &\|\omega^*(A_{1,t+1})-\omega^*(A_{1,t})\|_G\le\\ &\nu_1\|A_{1,t+1}-A_{1,t}\|_F=\nu_1 \mathrm{DOS}_{t+1}
\end{aligned} 
\end{equation*}
for some $0<\nu_1<\infty$. Using the above inequality to replace the last term in \eqref{eq:lem4eq1} yields:
\begin{equation}\label{eq:some-equation}
\begin{aligned}
    &\|\omega_{T+1}-\omega^*(A_{1,T+1})\|_G\\
    \le &\|\omega_{t_0}-\omega^*(A_{1,t_0})\|_G+\nu_1\sum_{t=t_0}^T\mathrm{DOS}_{t+1}
\end{aligned}
\end{equation}

Noticing that $\omega_{t_0}$, $\omega^*(A_{1,t_0})$, and $\omega^*(A_{1,T+1})$ are all bounded, there exists $0<\beta_1<\infty$ such that 
\begin{equation}\label{eq:lem4eq3}
    \|\omega^*(A_{1,T+1})\|_G+\|\omega_{t_0}-\omega^*(A_{1,t_0})\|_G\le \beta_1
\end{equation}
Using the triangle inequality for \eqref{eq:some-equation} and considering \eqref{eq:lem4eq3}, we obtain
\begin{equation}
    \|\omega_{T+1}\|_G\le 
    \beta_1+\nu_1\sum_{t=t_0}^T\mathrm{DOS}_{t+1}
\end{equation}

From the definition of $\omega$ and $\|\cdot\|_G$, we deduce that \eqref{eq:lem4} holds for all $T\ge t_0$. This concludes the proof.
\end{proof}

We then provide the proof for Theorem \ref{th:opt-gap}.

\begin{proof}
Under the new variables $s$, $z$ and $\lambda$,  \textbf{Steps 3}, \textbf{4} and \textbf{5} can be expressed as
\begin{equation}\label{eq:new-primal-dual}
\begin{aligned}
    z_{t+1}=&\mathop{\arg\min}_{z}~f_1(z)+\frac{\rho}{2}\left\|A_{1,t}z+A_2s_{t}-b-\frac{\lambda_t}{\rho}\right\|_F^2\\
    s_{t+1}=&\mathop{\arg\min}_s f_2(s)+\frac{\rho}{2}\left\|A_{1,t}z_{t+1}+A_2s-b-\frac{\lambda_t}{\rho}\right\|_F^2\\
    \lambda_{t+1}=&\lambda_t-\rho\left(A_{1,t}z_{t+1}+A_2s_{t+1}-b\right)
\end{aligned}
\end{equation}
From the optimality condition of \eqref{eq:new-primal-dual}, we have that
\begin{equation*}
\begin{aligned}
    &(z-z_{t+1})^\top \left[\partial f_1(z_{t+1})-A_{1,t}^\top\lambda_t\right.\\
   &~~~~~\left.+\rho A_{1,t}^\top \left(A_{1,t}z_{t+1}+A_2s_t-b\right)\right] \ge 0,\forall z\\
&(s-s_{t+1})^\top \left[\partial f_2(s_{t+1})-A_{2}^\top\lambda_t\right.\\
   &~~~~~\left.+\rho A_{2}^\top \left(A_{1,t}z_{t+1}+A_2s_{t+1}-b\right)\right] \ge 0,\forall s\\
\end{aligned}
\end{equation*}
Replacing $\lambda_{t}$ by $\lambda_{t+1}+\rho(A_{1,t}z_{t+1}+A_2s_{t+1}-b)$ in the above inequalities, we obtain the following
\begin{equation}\label{eq:optmality-condition}
\begin{aligned}
        &(z-z_{t+1})^\top \left[\partial f_1(z_{t+1})-A_{1,t}^\top\lambda_{t+1}\right.\\
   &~~~~~\left.+\rho A_{1,t}^\top A_2\left(s_{t}-s_{t+1}\right)\right] \ge 0,\forall z\\
&(s-s_{t+1})^\top \left[\partial f_2(s_{t+1})-A_{2}^\top\lambda_{t+1}\right]\ge 0,\forall s  \end{aligned}
\end{equation}
Using the convexity of $f_1(\cdot)$ and $f_2(\cdot)$, we have
\begin{equation}\label{eq:convexity}
\begin{aligned}
     &f_1(z_{t+1})-f_1(z_T^*)+f_2(s_{t+1})-f_2(s_T^*)\\
     \le &\partial f_1(z_{t+1})^\top (z_{t+1}-z_T^*)+\partial f_2(s_{t+1})^\top (s_{t+1}-s_T^*)
    \end{aligned}
\end{equation}
for some $z_T^*$ and $s_T^*$. Letting $z=z_T^*$ and $s=s_T^*$ in \eqref{eq:optmality-condition}, and using the convexity condition \eqref{eq:convexity}, we obtain
\begin{equation}\label{eq:??}
    \begin{aligned}
        &f_1(z_{t+1})-f_1(z_T^*)+f_2(s_{t+1})-f_2(s_T^*)\\
        \le &(-A_{1,t}^\top \lambda_{t+1}+\rho A_{1,t}^\top A_2(s_t-s_{t+1}))^\top (z_T^*-z_{t+1})\\
        &+(-A_2^\top \lambda_{t+1})^\top (s_T^*-s_{t+1})
    \end{aligned}
\end{equation}
Now consider the term $(-A_2^\top \lambda_{t+1})^\top (s_T^*-s_{t+1})$. Using the relations $A_2s_T^*=b-A_{1,T}z_T^*$ and $A_2s_{t+1}=\frac{1}{\rho}(\lambda_t-\lambda_{t+1})-(A_{1,t}z_{t+1}-b)$, it holds that
\begin{equation*}
\begin{aligned}
    &(-A_2^\top \lambda_{t+1})^\top (s_T^*-s_{t+1})\\
    =&\frac{1}{\rho}(\lambda_t-\lambda_{t+1})^\top \lambda_{t+1}+\lambda_{t+1}^\top (A_{1,T}z_T^*-A_{1,t}z_T^*)
\end{aligned}
\end{equation*}
Substituting this into \eqref{eq:??}, we obtain
\begin{equation*}
    \begin{aligned}
        &f_1(z_{t+1})-f_1(z^*)+f_2(s_{t+1})-f_2(s^*)\\
        \le &\frac{1}{\rho}(\lambda_{t}-\lambda_{t+1})^\top \lambda_{t+1}+\lambda_{t+1}^\top (A_{1,T}z^*-A_{1,t}z^*)\\
        &+\rho \left[(-A_2s_{t+1})-(-A_2s_{t})\right]^\top \left[(A_{1,t}z^*-b)-(A_{1,t}z_{t+1}-b)\right]
    \end{aligned}
\end{equation*}
Using the transformation
\begin{equation*}
\begin{aligned}
    &\left[(-A_2s_{t+1})-(-A_2s_{t})\right]^\top \left[(A_{1,t}z^*-b)-(A_{1,t}z_{t+1}-b)\right]\\
    =&\frac{1}{2}\left(\|A_{1,t}z_{t+1}+A_2s_{t+1}-b\|_F^2+\|A_{1,t}z^*+A_2s_t-b\|_F^2\right)\\
    -&\frac{1}{2}\left(\|A_{1,t}z^*+A_2s_{t+1}-b\|_F^2+\|A_{1,t}z_{t+1}+A_2s_t-b\|_F^2\right)
\end{aligned}
\end{equation*}
and $\|A_{1,t}z_{t+1}+A_2s_{t+1}-b\|_F^2=\frac{1}{\rho^2}\|\lambda_t-\lambda_{t+1}\|_F^2$, it holds
\begin{equation*}
    \begin{aligned}
        &f_1(z_{t+1})-f_1(z_T^*)+f_2(s_{t+1})-f_2(s_T^*)\\
        \le& \frac{1}{\rho}(\lambda_{t}-\lambda_{t+1})^\top \lambda_{t+1}+\\
        &\lambda_{t+1}^\top (A_{1,T}z_T^*-A_{1,t}z_T^*)+\frac{1}{2\rho}\|\lambda_t-\lambda_{t+1}\|_F^2\\
        +&\frac{\rho}{2}\left(\|A_{1,t}z_T^*+A_2s_t-b\|_F^2-\|A_{1,t}z_T^*+A_2s_{t+1}-b\|_F^2\right)
    \end{aligned}
\end{equation*}
The inequality is obtained after eliminating the term $-\frac{1}{2}\|A_{1,t}z_{t+1}+A_2s_t-b\|_F^2$. 

Thus, for any $\lambda$, it holds that
\begin{equation}\label{eq:!!}
\begin{aligned}
    &f_1(z_{t+1})-f_1(z_T^*)+f_2(s_{t+1})-f_2(s_T^*)\\
    &-\lambda^\top (A_{1,t}z_{t+1}+A_2s_{t+1}-b)\\
    \le &\frac{1}{\rho}(\lambda_{t+1}-\lambda)^\top (\lambda_t-\lambda_{t+1})\\
    &+\lambda_{t+1}^\top (A_{1,T}-A_{1,t})z_T^*+\frac{1}{2\rho}\|\lambda_t-\lambda_{t+1}\|_F^2\\
    +&\frac{\rho}{2}\left(\|A_{1,t}z_T^*+A_2s_t-b\|_F^2-\|A_{1,t}z_T^*+A_2s_{t+1}-b\|_F^2\right)
\end{aligned}
\end{equation}
Considering $\frac{\rho}{2}(\|A_{1,t}z_T^*+A_2s_t-b\|_F^2-\|A_{1,t}z_T^*+A_2s_{t+1}-b\|_F^2)$, it holds that 
\begin{equation}\label{eq:noname}
    \begin{aligned}
        &\|A_{1,t}z_T^*+A_2s_t-b\|_F^2-\|A_{1,t}z_T^*+A_2s_{t+1}-b\|_F^2\\
        =&\|A_{1,t-1}z^*+A_2s_t-b+(A_{1,t}-A_{1,t-1})z^*\|_F^2\\
        &-\|A_{1,t}z_T^*+A_2s_{t+1}-b\|_F^2\\
        \le &\|A_{1,t-1}z_T^*+A_2s_t-b\|_F^2-\|A_{1,t}z_T^*+A_2s_{t+1}-b\|_F^2\\
        &+\|(A_{1,t}-A_{1,t-1})z_T^*\|_F^2\\
        &+2\|A_{1,t-1}z_T^*+A_2s_t-b\|_F\| (A_{1,t}-A_{1,t-1})z_T^*\|_F
    \end{aligned}
\end{equation}
for any $t\ge 1$. 
The second equality is obtained by expanding the first term, and the inequality is obtained by using the submultiplicativity of vector norms.
Applying the above inequality and using
\begin{equation*}
\begin{aligned}
    &\lambda^\top (A_{1,t}z_{t+1}+A_2s_{t+1}-b)\\
    =&\lambda^\top (A_{1,T}^{\mathrm{nl}}z_{t+1}+A_2s_{t+1}-b)+\lambda^\top (A_{1,t}-A_{1,T}^{\mathrm{nl}})z_{t+1}
\end{aligned}
\end{equation*}
in \eqref{eq:!!}, we obtain
\begin{equation}\label{eq:long-ieq}
    \begin{aligned}
    &f_1(z_{t+1})-f_1(z_T^*)+f_2(s_{t+1})-f_2(s_T^*)\\
    &-\lambda^\top (A_{1,T}^{\mathrm{nl}}z_{t+1}+A_2s_{t+1}-b)\\
    \le &\frac{1}{2\rho}\left(\|\lambda-\lambda_t\|_F^2-\|\lambda-\lambda_{t+1}\|_F^2\right)+\lambda_{t+1}^\top (A_{1,T}^{\mathrm{nl}}-A_{1,t})z_T^*\\
    &+\frac{\rho}{2}\left(\|A_{1,t-1}z_T^*+A_2s_t-b\|_F^2-\|A_{1,t}z_T^*+A_2s_{t+1}-b\|_F^2\right)\\
    &+\frac{\rho}{2}\left(\|(A_{1,t}-A_{1,t-1})z_T^*\|_F^2\right.\\
        &\left.+2\|A_{1,t-1}z_T^*+A_2s_t-b\|_F\| (A_{1,t}-A_{1,t-1})z_T^*\|_F\right)\\
        &+\|\lambda\|_F \|(A_{1,t}-A_{1,T}^{\mathrm{nl}})z_{t+1}\|_F
    \end{aligned}
\end{equation}
Define 
\begin{equation}
    \overline{z}_T:=\frac{1}{T-t_0+1}\sum_{t=t_0}^T z_t, \quad \overline{s}_T:=\frac{1}{T-t_0+1}\sum_{t=t_0}^Ts_t
\end{equation}
Then, summing \eqref{eq:long-ieq} over $t=t_0,\ldots,T$ yields
\begin{equation}\label{eq:??!!}
    \begin{aligned}
        &f_1(\overline{z}_T)-f_1(z_T^*)+f_2(\overline{s}_T)-f_2(s_T^*)-\lambda^\top ( A_{1,T}^{\mathrm{nl}}\overline{z}_T+A_2\overline{s}_T-b)\\
        &\le \frac{1}{2\rho(T-t_0+1)}\|\lambda-\lambda_{t_0}\|_F^2+\\
        &\frac{\rho}{2(T+1)}\|A_{1,t_0}z_T^*+A_2s_{t_0}-b\|_F^2\\
        &+\frac{1}{T-t_0+1}\left\{\sum_{t=t_0+1}^T\left[\vphantom{\sum_{t=t_0}^T}\underbrace{\|\lambda_{t+1}\|_F\|A_{1,T}^{\mathrm{nl}}-A_{1,t}\|_F\|z_T^*\|_F}_{\text{\ding{172}}}\right.\right.\\
        &+\underbrace{\frac{\rho}{2}\|(A_{1,t}-A_{1,t-1})z_T^*\|_F^2}_{\text{\ding{173}}}+\underbrace{\|\lambda\|_F\|A_{1,T}^{\mathrm{nl}}-A_{1,t}\|_F\|z_{t+1}\|_F}_{\text{\ding{174}}}\\
        &\left.+\underbrace{\rho \|A_{1,t-1}z_T^*+A_2s_t-b\|_F\| A_{1,t}-A_{1,t-1}\|_F \|z_T^*\|_F\vphantom{\sum_{t=0}^T}}_{\text{\ding{175}}}\right]\\
        &+\|A_{1,t_0}z_T^*+A_2s_{t_0}-b\|_F^2+\|\lambda_{1}\|_F\|A_{1,T}^{\mathrm{nl}}-A_{1,t_0}\|_F\|z_T^*\|_F\\
        &\left.+\|\lambda\|_F\|A_{1,T}^{\mathrm{nl}}-A_{1,t_0}\|_F\|z_T^*\|_F\vphantom{\sum_{t=t_0}^T}\right\}
    \end{aligned}
\end{equation}

Here, we have used the fact that 
\begin{equation}
    f_1(\overline{z}_T)+f_2(\overline{s}_T)=\sum_{t=t_0}^T \left(f_1(z_t)+f_2(s_t)\right)
\end{equation}
This is because $f_1$ is linear, and $f_2(s_t)=0$ for all $t\ge t_0$ because of the projection in \textbf{Step 4} of Algorithm \ref{al:online-control-first-order}.

The sum on the right-hand side of \eqref{eq:long-ieq} is taken from $t_0+1$ to $T$ because of the inequality \eqref{eq:noname}, which holds for $t\ge t_0+1$. For $t=t_0$, the terms are presented outside the summation $\sum_{t=t_0+1}^T$. To analyze the right-hand side of the inequality in \eqref{eq:??!!}, we separately consider the terms \ding{172}, \ding{173}, \ding{174} and \ding{175} inside $\frac{1}{T-t_0+1}\sum_{t=t_0+1}^T$.

\ding{172}: 
The term $\| {A}_{1,T}^{\mathrm{nl}}-A_{1,t}\|_F$ can be bounded by
\begin{equation*}
\begin{aligned}
    &\|{A}_{1,T}^{\mathrm{nl}}-A_{1,t}\|_F\le \|{A}_{1,T}^{\mathrm{nl}}-A_{1,T}\|_F+\|A_{1,T}-A_{1,t}\|\\
    &\le\mathrm{SRG}_T+\beta_2
\end{aligned}
\end{equation*}
where $\beta_2:=2\max_{A_1\in \mathcal{A}} \|A_1\|_F$. Under the assumption that $\mathcal{A}$ is compact and $A_{1,t}\in\mathcal{A},\forall t\ge t_0$, we have $\beta_2<\infty$. 
Using Lemma \ref{lem:boundedness} for $\|\lambda_{t+1}\|_F$, we obtain the following inequality
\begin{equation}\label{eq:bound172}
    \begin{aligned}
        &\|\lambda_{t+1}\|_F\|A_{1,T}^{\mathrm{nl}}-A_{1,t}\|_F\|z_T^*\|_F\le  \\
        &\left(\beta_1+\nu_1\sum_{k=t_0}^t\mathrm{DOS}_k\right)\left(\mathrm{SRG}_T+\beta_2\right)\beta_3
    \end{aligned}
\end{equation}
where $\beta_3:=\|z^*_T\|_F$.

\ding{173}: We have
\begin{equation}\label{eq:bound173}
    \begin{aligned}
        &\frac{\rho}{2}\|(A_{1,t}-A_{1,t-1})z_T^*\|_F^2\le \frac{\rho}{2}\mathrm{DOS}_t^2 \beta_2^2
    \end{aligned}
\end{equation}

\ding{174}: This term appears similarly to \ding{173}; the only difference is between $\|\lambda_{t+1}\|_F$ and $\|\lambda\|_F$. Following \eqref{eq:bound172}, the bound for this term is given by
\begin{equation}\label{eq:bound174}
    \begin{aligned}
        &\|\lambda\|_F\|A_{1,T}^{\mathrm{nl}}-A_{1,t}\|_F\|z_{t+1}\|_F\le  \\
        &\|\lambda\|_F\left(\mathrm{SRG}_T+\beta_2\right)\left(\beta_1+\nu_1\sum_{k=t_0}^t\mathrm{DOS}_k\right)
    \end{aligned}
\end{equation}
We note here that the bound in \eqref{eq:??!!} is established for any $\lambda$; the selection of $\lambda$ will be clear in the sequel.

\ding{175}: Consider the first part:
\begin{equation}\label{eq:bound175}
    \begin{aligned}
        &\rho\|A_{1,t-1}z_T^*+A_2s_t-b\|_F\\
        =&\rho\|A_{1,T}z_T^*+A_2s_T^*-b+(A_{1,t-1}-A_{1,T})z^*+A_2(s_t-s_T^*)\|_F\\
        \le &\rho\|(A_{1,t-1}-A_{1,T})z_T^*\|_F+\rho\|A_2(s_t-s_T^*)\|_F\\
        \le &\rho \beta_2\beta_3+\left(\beta_1+\nu_1\sum_{k=t_0}^{t-1} \mathrm{DOS}_{k}\right)+\beta_4  \end{aligned}
\end{equation}
Here, the first inequality is obtained via the triangle inequality and the fact that $A_{1,T}z_T^*+A_2s_T^*-b=0$. For the other quantities, $\rho \|(A_{1,t-1}-A_{1,T})z_T^*\|_F\le \rho \beta_2\beta_3$ because $\|A_{1,t-1}-A_{1,T}\|_F\le 2\max_{A_1\in\mathcal{A}}\|A_1\|_F=\beta_2$ and $\|z_T^*\|_F=\beta_3$. The bound for $\rho \|A_2(s_t-s_T^*)\|_F$ is obtained as $\rho \|A_2(s_t-s_T^*)\|_F\le \rho \|A_2s_t\|_F+\rho \|A_2(s_T^*)\|_F$. Let $\beta_4:=\rho \|A_2(s_T^*)\|_F$. We then obtain the final result \eqref{eq:bound175}. 

We then consider the second part $\|A_{1,t}-A_{1,t-1}\|_F\|z^*_T\|_F$ in \ding{175}. An upper bound can be obtained similarly to \ding{173} as $\mathrm{DOS}_t \beta_2$. Combining this with \eqref{eq:bound175}, the final upper bound for \ding{175} is expressed as
\begin{equation}\label{eq:bound175final}
    \begin{aligned}
\mathrm{DOS}_t\beta_2\left(\rho \beta_2\beta_3+\left(\beta_1+\nu_1\sum_{k=t_0}^{t-1} \mathrm{DOS}_{k}\right)+\beta_4\right)        
    \end{aligned}
\end{equation}

We then consider the remaining terms
in \eqref{eq:??!!}. It is worth mentioning that  \eqref{eq:??!!} holds for any $\lambda$, and $\lambda$ is independent of $\lambda_t$ in the inequality. Consider the specific choice
\begin{equation}\label{eq:lambda-choice}
    \lambda=-\eta ( A_{1,T}^{\mathrm{nl}}\overline{z}_T+A_2\overline{s}_T-b)/\|A_{1,T}^{\mathrm{nl}}\overline{z}_T+A_2\overline{s}_T-b\|_F
\end{equation}
with some $\eta>0$. Then, it holds that $\|\lambda\|_F=\eta$. With this in mind, we can obtain an upper bound $\beta_5/(T+1)$ with $0<\beta_5<\infty$ for these terms. Substituting this, \eqref{eq:bound172}, \eqref{eq:bound173}, \eqref{eq:bound174}, \eqref{eq:bound175final}, and \eqref{eq:lambda-choice} into \eqref{eq:??!!}, we obtain
\begin{equation}\label{eq:heavy}
    \begin{aligned}
        &f_1(\overline{z}_T)-f_1(z_T^*)+f_2(\overline{s}_T)-f_2(s_T^*)+\eta\| A_{1,T}^{\mathrm{nl}}\overline{z}_T+A_2\overline{s}_T-b\|_F\\
        &\le \underbrace{\frac{D_1+D_2\sum_{t=t_0}^T\sum_{k=t_0}^t \mathrm{DOS}_k}{T-t_0+1}+D_3\mathrm{SRG}_T}_{:=S_T}\\
    \end{aligned}
\end{equation}
where $0 < D_1 < \infty$ accounts for the constant terms, and $0 < D_2, D_3 < \infty$ correspond to scaling factors associated with the accumulation of $\mathrm{DOS}_t$ and $\mathrm{SRG}_T$, respectively. Without loss of generality, we assume that $\mathrm{SRG}_T$ is bounded, which allows us to separate its contribution from $\sum_{k=t_0}^t \mathrm{DOS}_k$ in \eqref{eq:bound174}. Similarly, we assume that $\mathrm{DOS}_t$ is bounded, so that its interaction with $\sum_{k=t_0}^t \mathrm{DOS}_k$ can be handled analogously in \eqref{eq:bound175final}. These assumptions simplify the regret bound and make the dependence on $\mathrm{DOS}_t$ and $\mathrm{SRG}_t$ explicit. It can be seen that \eqref{eq:heavy} consists of three parts: one converges sublinearly, one depends on the accumulation of $\mathrm{DOS}_k$, and the other one depends on the noise $\mathrm{SRG}_t$. 

We have now obtained an upper bound for $f_1(\overline{z}_T)+f_2(\overline{s}_T)-f_1(z_T^*)-f_2(s^*_T)$, despite the appearance of $\eta\|\hat A_{1,T}\overline{z}_T+A_2\overline{s}_T-b\|_F$. However, given that this term is always nonnegative, one meaningful upper bound can be derived by eliminating this term, i.e.,
\begin{equation}\label{eq:upper-bound}
    f_1(\overline{z}_T)-f_1(z_T^*)+f_2(\overline{s}_T)-f_2(s_T^*)\le S_T
\end{equation}

We then evaluate a lower bound. Consider dualizing the optimization problem \eqref{eq:dee-sdp2} with $A_{1,T}^{\mathrm{nl}}$:
    \begin{equation*}
        L(z,s,\lambda)=f_1(z)+f_2(s)+\lambda^\top( A_{1,T}^{\mathrm{nl}}z+A_2s-b)
    \end{equation*}
From weak duality, it holds that $L(\hat z_T^*,\hat s_T^*,\lambda)\le L(\hat z_T^*,\hat s_T^*,\hat \lambda_T^*)\le L(z,s,\hat \lambda_T^*)$, where $(\hat z_T^*,\hat s_T^*,\hat \lambda_T^*)$ is the optimal solution that corresponds to $A_{1,T}^{\mathrm{nl}}$, satisfying $A_{1,T}^{\mathrm{nl}}\hat z_T^*+A_2\hat s_T^*=b$. By taking $z=\overline{z}_T$ and $s=\overline{s}_T$, the rightmost inequality implies
\begin{equation}\label{eq:lowerbound1}
    f_1(\hat z_T^*)+f_2(\hat s_T^*)\le f_1(\overline{z}_T)+f_2(\overline{s}_T)+(\hat \lambda_T^*)^\top ( A_{1,T}^{\mathrm{nl}}\overline{z}_T+A_2\overline{s}_T-b)
\end{equation}

Again, using Lemma \ref{lem:smoothness}, we obtain 
\begin{equation}\label{eq:lowerbound2}
\begin{aligned}
    &|f_1(\hat z_T^*)-f_1(z_T^*)|=|c^\top(\hat z_T^*-z_T^*)|\\
    &\le \|c\|_F \|\omega^*(A_{1,T}^{\mathrm{nl}})-\omega^*(A_{1,T})\|_F \\
    &\le \|c\|_FL \|A_{1,T}^{\mathrm{nl}}-A_{1,T}\|_F\le \underbrace{\|c\|_F L }_{:=B_4} \mathrm{SRG}_T
\end{aligned}
\end{equation}
Substituting \eqref{eq:lowerbound2} into \eqref{eq:lowerbound1}, and considering the fact that $f_2(\hat s_T^*)=f_2(\bar s_T)=0$, we have 
\begin{equation}\label{eq:someineq2}
    \begin{aligned}
        &f_1(\overline{z}_T)+f_2(\overline{s}_T)-f_1(z_T^*)-f_2(s_T^*)\\
        \ge & -\|\hat \lambda_T^*\|_F\| A_{1,T}^{\mathrm{nl}}\overline{z}_T+A_2\overline{s}_T-b\|_F-B_4 \mathrm{SRG}_T
    \end{aligned}
\end{equation}
Comparing the upper bound in \eqref{eq:heavy} and the lower bound in \eqref{eq:someineq2}, and taking $\eta=\|\hat \lambda_T^*\|_F+1$, we obtain
\begin{equation}
    \| A_{1,T}^{\mathrm{nl}}\overline{z}_T+A_2\overline{s}_T-b\|_F\le S_T+D_4\mathrm{SRG}_T
\end{equation}
This inequality reveals the convergence rate of the primal residual under noiseless data.
Substituting this inequality into \eqref{eq:someineq2}, we get 
\begin{equation}\label{eq:lower-bound}
    f_1(\overline{z}_T)+f_2(\overline{s}_T)-f_1(z_T^*)-f_2(s_T^*)\ge -(\eta-1)S_T-\eta D_4\mathrm{SRG}_T
\end{equation}
Taking the maximum absolute value of the bounds in \eqref{eq:lower-bound} and \eqref{eq:upper-bound}, and considering the fact that $f_2(s_t)=f_2(s_T^*)=0$, we obtain \eqref{eq:opt-gap} for some $0<B_2,B_3,B_4<\infty$.

\end{proof}


\FloatBarrier
\bibliographystyle{ieeetr}
\bibliography{ref}

\end{document}